\documentclass[twocolumn,10pt]{IEEEtran}
\usepackage{braket}

\usepackage{booktabs}

\input{def.tex}
\usepackage{dsfont}
\usepackage{minibox}
\DeclareSymbolFont{matha}{OML}{txmi}{m}{it}
\DeclareMathSymbol{\varv}{\mathord}{matha}{118}
\usepackage{eurosym}
\usepackage{hhline,physics}

\usepackage{multicol}
\usepackage{algorithm}
\usepackage{graphicx}
\usepackage{textcomp}
\usepackage{caption,pifont}
\usepackage[multiple]{footmisc}

\usepackage{algorithmic}

\makeatletter
\renewcommand{\baselinestretch}{0.99}

\IEEEoverridecommandlockouts

\begin{document}
	\title{Achievable Rate Optimization for Stacked Intelligent Metasurface-Assisted Holographic MIMO Communications} 
	\author{Anastasios Papazafeiropoulos, Jiancheng An, Pandelis Kourtessis, Tharmalingam Ratnarajah, Symeon Chatzinotas \thanks{A. Papazafeiropoulos is with the Communications and Intelligent Systems Research Group, University of Hertfordshire, Hatfield AL10 9AB, U. K., and with SnT at the University of Luxembourg, Luxembourg. J. An is with the School of Electrical and Electronics Engineering, Nanyang Technological University, Singapore 639798. P. Kourtessis is with the Communications and Intelligent Systems Research Group, University of Hertfordshire, Hatfield AL10 9AB, U. K. T. Ratnarajah  is with he School of Engineering, Institute for Digital Communications, The University of Edinburgh, EH8 9YL Edinburgh, U.K. S. Chatzinotas is with the SnT at the University of Luxembourg, Luxembourg. A. Papazafeiropoulos was supported  by the University of Hertfordshire's 5-year Vice Chancellor's Research Fellowship.
		S. Chatzinotas   was supported by the National Research Fund, Luxembourg, under the project RISOTTI. E-mails: tapapazaf@gmail.com,  jiancheng.an@ntu.edu.sg,  p.kourtessis@herts.ac.uk, t.ratnarajah@ieee.org, symeon.chatzinotas@uni.lu.}}
	\maketitle\vspace{-1.7cm}
	\begin{abstract}	
		Stacked intelligent metasurfaces (SIM) is a revolutionary technology, which can outperform its single-layer counterparts by performing advanced signal processing relying on wave propagation. In this work, we exploit SIM to enable transmit precoding and receiver combining in holographic multiple-input multiple-output (HMIMO) communications, and we study the achievable rate by formulating a joint optimization problem of the  SIM phase shifts at both sides of the transceiver and the covariance matrix of the transmitted signal. Notably, we propose its solution by means of an iterative optimization algorithm that relies on the projected gradient method, and accounts for all optimization parameters simultaneously. We also obtain the step size guaranteeing the convergence of the proposed algorithm. Simulation results provide fundamental insights such the performance improvements compared to the single-RIS counterpart and conventional MIMO system. Remarkably, the proposed algorithm results in the same achievable rate as the alternating optimization (AO) benchmark but with a less number of iterations.
	\end{abstract}
	
	
	\section{Introduction}
	The need for sixth-generation (6G) cellular networks has appeared on the horizon of the wireless network evolution \cite{Letaief2019}. Their extreme targets require vast improvements on data rates and latency together with wider connectivity to cover the explosive proliferation of the Internet-of-Everything (IoE), such as virtual/augmented reality (VR/AR) and the increase in connected devices. In particular, the latter is expected to reach $ 500 $ million by $ 2030 $  \cite{Dang2020}. Various technologies such as massive multiple-input multiple-output (mMIMO) and millimeter wave (mmWave) communications that have been suggested in recent years that require low energy consumption and can achieve high data rates, but they concern transceiver features without being able to shape the propagation channel with its stochastic characteristics that limit the performance \cite{Andrews2014}.
	
	Two possible approaches that solve the aforementioned issues are the proposed 6G-enabled technologies, reconfigurable intelligent surfaces (RIS) \cite{DiRenzo2020,Wu2020,Papazafeiropoulos2021,Papazafeiropoulos2023c}, and holographic multiple-input multiple-output (HMIMO)
	Communications \cite{Huang2020}. Specifically, RIS and its various equivalents have been proposed to realize smart reconfigurable environments \cite{DiRenzo2020}. A RIS is a planner metasurface equipped with a large number of nearly passive elements that induce phase shifts and/or amplitude attenuation to the impinging waves through a smart controller. Also, optimization of the reflected signals enable us to control the interaction of the reflection characteristics with the surrounding objects. RIS performance gains have been studied under various system and channel setups, but most works have focused on single-input single-output (SISO) or multiple-input single-output (MISO) systems with single-antenna receivers \cite{Wu2019,Bjoernson2019b,Yang2020b,Zhao2020,Papazafeiropoulos2021,Mu2021,Papazafeiropoulos2023,Papazafeiropoulos2023a, Papazafeiropoulos2023b}, while the research on RIS-assisted MIMO assisted is limited \cite{Pan2020,Ye2020}. In particular, the study of the capacity limit of RIS-assisted MIMO systems requires the joint optimization of the RIS phase shifts and MIMO transmit covariance matrix \cite{Zhang2020a,Perovic2021}. It is worthwhile to mention that since RISs do not require active transmitter radio frequency (RF) chains, they can be densely implemented with low energy consumption and low cost \cite{Wu2020}. However, single-layer RIS designs are not capable of implementing advanced MIMO functionalities because of hardware limitations and suffer from severe path-loss attenuation.
	
	On the ground of mMIMO systems with its improvements on spectral and energy efficiencies and their other benefits such as reduced latency \cite{Marzetta2016}, HMIMO has emerged as a new concept describing possibly the next MIMO generation \cite{Huang2020,Wan2021}. For example, in \cite{Hu2018}, a large intelligent surface (LIS) including a massive number of elements it was shown that great improvements are expected. The fundamental limits of HMIMO systems were characterised in \cite{Pizzo2020} by suggesting correlated random Gaussian fading for the far field. In the case of arbitrary scattering environments a Fourier plane-wave series-based expansion of the HMIMO channel response was developed in \cite{Pizzo2022}. In \cite{Demir2022}, a channel estimation scheme was proposed for arbitrary spatial correlation matrices. However, HMIMO have the disadvantages of excessive hardware cost and energy consumption because they are implemented by a large number of active components.
	
	Recently, not only stacked intelligent metasurfaces (SIMs) have been proposed where multiple surfaces are cascaded \cite{Liu2022,An2023a}, but SIMs were integrated with the transceiver to implement HMIMO communications in \cite{An2023}. It was shown that a SIM can implement signal processing in the electromagnetic 	(EM) wave regime. The motivation behind this work was to substitute the exprensive active elements at the tranceiver by exploiting the technology of programmable metasurfaces. In \cite{Hu2018,Huang2020}, single-layer surfaces were used but the multilayer surface design is more advantageous for enhancing the spatial-domain gain by forming diverse waveforms with high accuracy. 
	
	From theoretical and practical standpoints, it is crucial to optimise the achievable rate for RIS-assisted systems. Hence, various optimization methods have been suggested in prior works, which aim to find near-optimal solutions obeying to reasonable run time and computational complexity. Most of these works considered single-antenna receive devices and relied on the alternating optimization (AO) method which optimises the transmit beamformer and the RIS phase shifts in an alternating way \cite{Zhao2020,Papazafeiropoulos2021}. For example, in \cite{Zhao2020}, the gradient method was employed in an AO fashion to optimize the phase shifts. Other examples of AO-based works on RIS-assisted MIMO communication are \cite{Oezdogan2020,Perovic2020,Zhang2020a,Perovic2021}. In \cite{Oezdogan2020}, a RIS was optimized to increase the rank of the channel matrix. In \cite{Perovic2020}, the optimization took place in an indoor mmWave environment. Moreover, in \cite{Zhang2020a}, the achievable rate of RIS-assisted multi-stream MIMO was maximized by using the AO method. However, AO-based methods require possibly many iterations to converge, which increase with the size of the RIS. Note that this is the case, where a RIS is more beneficial in practice. Contrary to this background, a joint optimization of the RIS elements and the transmit covariance matrix in \cite{Perovic2021}, where an iterative projected gradient method was proposed.

	\textit{Contributions}: Motivated by the above observations, the topic of this paper concerns the study of SIM-enabled HMIMO systems by optimising simultaneously all parameters of the joint optimization problem to reduce the convergence time compared to an AO approach. Contrary to \cite{An2023}, which considers a full analog SIM-enabled HMIMO architecture by means of an AO approach we focus on a more general hybrid digital and  wave design that also applies a more efficient algorithm optimizing all relevant parameters simultaneously. Compared to \cite{Perovic2021} and \cite{Papazafeiropoulos2023}, which assumed a conventional RIS-assisted system and a STAR-RIS system, where two parameters are optimized simultaneously, we consider a general SIM-enabled system, where we optimize three key parameters simultaneously. Our main contributions are summarised as follows.
	\begin{itemize}
		\item We maximise the achievable rate of a multi stream HMIMO system equipped with a SIM at the transmitter and a SIM at the receiver. To this end, we formulate the joint optimization problem of the transmit covariance matrix, and the RIS phase shift values of each surface at the transmitter and the receiver SIMs.
\item	We propose an iterative projected gradient approach, which solves the underlying nonconvex problem. Also, we derive the gradients and the projection expressions for all parameters in closed-forms. Moreover, we show that the proposed approach converges to a critical point. 
\item	We determine the appropriate step size that makes the proposed algorithm to converge by deriving first the Lipschitz constant.
\item	Simulation and analytical results coincide and show that both the proposed approach and the AO method result in the same rate but our approach achieves it with a substantially lower number of iterations. Furthermore, the proposed approach has remarkably lower computational complexity with respect to the AO method.
		\end{itemize}
	\textit{Paper Outline}: The structure of this papers follows. Section~\ref{System} presents the system model and the problem formulation of a SIM-assisted HMIMO system.  Section~\ref{RateOptimization} presents the simultaneous optimization of the achievable rate with respect to all its parameters. Section \ref{convergence} provides the convergence and complexity analyses. In Section~\ref{Numerical}, we provide the numerical results,  and Section~\ref{Conclusion} concludes the paper.

\textit{Notation}: Vectors and matrices are denoted by boldface lower and upper case symbols, respectively. The notations $(\cdot)^\T$, $(\cdot)^\H$, and $\tr\!\left( {\cdot} \right)$ describe the transpose, Hermitian transpose, and trace operators, respectively. Moreover, the notations $ \arg\left(\cdot\right) $ and $\EE\left[\cdot\right]$  express the argument function and the expectation  operator, respectively. The notation  $\diag\left(\bA\right) $ describes a vector with elements equal to the  diagonal elements of $ \bA $, the notation  $\diag\left(\bx\right) $ describes a diagonal  matrix whose elements are $ \bx $, while  $\bb \sim \cC\cN{(\b0,\mathbf{\Sigma})}$ describes a circularly symmetric complex Gaussian vector with zero mean and a  covariance matrix $\mathbf{\Sigma}$. 	
	
	\section{System Model and Problem Formulation}\label{System}
	\subsection{System Model}
We consider a SIM-assisted HMIMO, where the transmitter and the receiver have $ N_{t} $ and $ N_{r} $ antennas, respectively,
while each of them is assisted by a SIM as shown in Fig. \ref{Algoa1}. Specifically, each SIM consists of a closed vacuum container having several stacked metasurface layers \cite{Liu2022}. The operation requires a customized field programmable gate array (FPGA), or generally a smart controller, which can adjust the phase shift of the EM waves impinging on each meta-atom. Hence, a customized
spatial waveform shape at the output of each metasurface layer is automatically produced as the transmit signals propagate through the SIM. The EM waves can be transmitted from the output metasurface of the transmitter SIM into the ether, and then, acquired by the receiver SIM, i.e., HMIMO communication can be supported while by leveraging the SIM-based analog beamforming to approach its digital counterpart. In particular, the transmitter SIM plays the role of the precoder by sending the information-bearing EM wave through the ether to the receiver SIM, which can combine the received EM to recover the transmitted signal. In other words, precoding and combining take place partially in the wave domain \cite{An2023}.
	
\begin{figure}
	\begin{center}
		\includegraphics[width=0.8\linewidth]{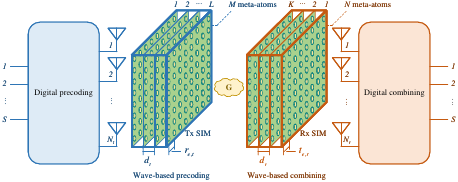}
		\caption{A SIM-assisted HMIMO system. }
		\label{Fig01}
	\end{center}
\end{figure}
	
According to Fig. \ref{Fig01}, which illustrates the SIM-assisted
HMIMO system supporting precoding and combining in the wave domain, we rely on the SIM design proposed in \cite{An2023}. Specifically, we denote $ L $ the number of metasurface layers
at the transmitter and $ K $ the number of metasurface layers at the receiver, while $ \mathcal{L}=\{1,\ldots,L\} $ and $\mathcal{K}=\{1,\ldots,K\} $ represent their sets, respectively. Without any loss of generality, we assume that each metasurface layer at the transmitter and the receiver consists of an identical number of meta-atoms.	In particular, we denote $ M $ and $ N $ the number of meta-atoms on each metasurface layer at the transmitter SIM and receiver SIM, respectively. The respective sets are denoted as $ \mathcal{M}=\{1,\ldots,M\} $ and $ \mathcal{N}=\{1,\ldots,N\} $.
 On this ground, we denote $ \theta_{m}^{l}\in [0,2\pi), m \in \mathcal{M}, l \in \mathcal{L} $ the phase shift by  meta-atom  $ m $ on the  transmit metasurface layer $ l $ with $ \phi_{m}^{l} =e^{j \theta_{m}^{l}}$ being the respective transmission coefficient. The transmission coefficient matrix, associated with the $l$-th transmit layer, is denoted by $ \bPhi^{l}=\diag(\bphi^{l})\in \mathbb{C}^{M \times M} $, where $ \bphi^{l} =[\phi^{l}_{1}, \dots, \phi^{l}_{M}]^{\T}\in \mathbb{C}^{M \times 1}$. Similarly, we denote $ \xi_{n}^{k}\in [0,2\pi), n \in \mathcal{N}, k \in \mathcal{K} $ the phase shift by meta-atom  $ n $ on the  receive metasurface layer  $ k $ with $ \psi_{n}^{k}=e^{j \xi_{n}^{k}} $ being the corresponding transmission coefficient. The $k$-th receive coefficient matrix is denoted by $ \bPsi^{k}=\diag(\bpsi^{k})\in \mathbb{C}^{M \times N} $, where $ \bpsi^{k} =[\psi^{k}_{1}, \dots, \psi^{k}_{M}]^{\T}\in \mathbb{C}^{N \times 1}$.
\begin{remark}
In this work, we assume continuously-adjustable phase shifts and constant modulus equal to $ 1 $ to assess the performance of SIM-assisted HMIMO communications while maximizing the achievable rate as in \cite{Wu2019}. Practical issues such as the assumption of coupled phase and magnitude \cite{Abeywickrama2020} and the consideration of discrete phase shifts \cite{Liu2022} will be studied in future work.
	\end{remark}
\begin{remark}
	Note that our hybrid digital and wave architecture employing multiple-layer SIM significantly outperforms the conventional hybrid benchmarks whose performance is constrained by analog components, e.g., constant- modulus phase shifters. Remarkably, our hybrid architecture may approach the performance of an all-digital system, while  the number of RF chains reduces from $M$ to $N_t$.
\end{remark}

With all metasurface layers having an isomorphic lattice arrangement \cite{Liu2022}, we model each surface as a uniform planar array, where
the element spacing between the $ \tilde{m} $th and $ m $th meta-atoms on the same transmit metasurface is given by \cite{An2023}
\begin{align}
	r_{m,\tilde{m}}=r_{e,t}\sqrt{(m_{z}-\tilde{m}_{z})^{2}+(m_{x}-\tilde{m}_{x})^{2}}
\end{align}
with $ r_{e,t} $ expressing the element spacing between adjacent meta-atoms on the same transmit metasurface. Moreover, $ m_{x} $ and $ m_{z} $ correspond to the $ m $th meta-atom along the $ x $-axis and the $ z $-axis, respectively, which are given by
\begin{align}
 m_{x}=\!\!\!\!\!\mod(m-1, m_{\mathrm{max}})+1,~~~m_{z}=\lceil m/m_{\mathrm{max}}\rceil
\end{align}
with $ m_{\mathrm{max}} $ being the highest number of meta-atoms included on each row of the transmit surface. In a similar way, the element spacing between the $ n $th meta-atom and the $ \tilde{n} $th one on the same receive metasurface can be described as
\begin{align}
	t_{\tilde{n},n}=t_{e,r}\sqrt{(\tilde{n}_{x}-n_{x})^{2}+(\tilde{n}_{z}-n_{z})^{2}},
\end{align}
 where $ t_{e,r} $ expresses the element spacing between adjacent meta-atoms on the same receive metasurface with $ n_{x} $ and $ n_{z} $ describing the indices of the $ n $th meta-atom along the $ x $-axis and the $ z$-axis, respectively. These indices are given by
 \begin{align}
 	n_{x}=\!\!\!\!\!\mod(n-1, n_{\mathrm{max}})+1, ~~~n_{z}&=\lceil n/n_{\mathrm{max}}\rceil	,
 \end{align}
 where $ n_{\mathrm{max}} $ is the maximum number of meta-atoms on each row of the receive metasurface. For the sake of simplicity, we assume square surfaces at both the transmitter and receiver sides, i.e., $ M=m_{\mathrm{max}}^{2} $ and $ N=n_{\mathrm{max}}^{2} $.
 
 By assuming uniform spacing among all surfaces and that all surfaces are parallel for the sake of simplicity, we define the transmission distance from   meta-atom $ \tilde{m} $ on the  transmit metasurface $ (l-1) $ to  meta-atom $ m $ on the  transmit metasurface $ l $ as
 \begin{align}
 	r_{m,\tilde{m}}^{l}=\sqrt{r_{m,\tilde{m}}^{2}+d_{t}^{2}}, ~l\ \in \mathcal{L}/\{1\},
 \end{align}
where $ d_{t}={D_{t}}/{L} $ is the spacing between any two adjacent metasurfaces at the transmitter SIM with $ D_{t} $ describing the thickness of the transmitter SIM.

Similarly, we define the distance from   meta-atom $ \tilde{n} $ on the  receive metasurface $ k $ to  meta-atom $ \tilde{n} $ on the receive metasurface $ (k-1) $ as
\begin{align}
	t_{\tilde{n},n}^{k}=\sqrt{d_{r}^{2}+t_{\tilde{n},n}^{2}},~ k\ \in \mathcal{K}/\{1\},
\end{align}
where $ d_{r}={D_{r}}/{L} $ is the spacing between any two adjacent metasurfaces at the receiver SIM with $ D_{r} $ describing the thickness of the receiver SIM.

In addition, by assuming that the centers of the transmit and receive antenna arrays are aligned with the centers of all metasufaces while both antenna arrays are arranged in a uniform linear array with element spacing $ \lambda/2 $, the distances from the $ s $th source to the $ m $th meta-atom on the input metasurface of the transmitter SIM and from $ n $th meta-atom on the output metasurface of the receiver SIM to the $ s $th destination are provided by \eqref{SIM1} and \eqref{SIM2}, respectively.

\begin{figure*}
	\begin{align}
	r_{m,s}^{1}&=\sqrt{\left[\left(m_{z}-\frac{m_{\mathrm{max}+1}}{2}\right)r_{e,t}-\left(s-\frac{N_{t}+1}{2}\right)\frac{\lambda}{2}\right]^{2}+\left(m_{x}-\frac{m_{\mathrm{max}}+1}{2}\right)^{2}r_{e,t}^{2}+d_{t}^{2}},\label{SIM1}\\	t_{s,n}^{1}&=\sqrt{\left[\left(n_{z}-\frac{n_{\mathrm{max}+1}}{2}\right)t_{e,r}-\left(s-\frac{N_{r}+1}{2}\right)\frac{\lambda}{2}\right]^{2}+\left(\frac{n_{\mathrm{max}}+1}{2}-n_{x}\right)^{2}t_{e,r}^{2}+d_{r}^{2}}.\label{SIM2}
	\end{align}
	\hrulefill
\end{figure*}

\subsection{Channel Model}
Regarding the transmission coefficient from  meta-atom  $ \tilde{m} $ on the  transmit metasurface layer $ (l-1) $ to  meta-atom  $ m $ on the  transmit metasurface layer $ l $, provided by the Rayleigh-Sommerfeld diffraction theory \cite{Lin2018}, it is given by
\begin{align}
	w_{m,\tilde{m}}^{l}=\frac{A_{t}cos x_{m,\tilde{m}}^{l}}{r_{m,\tilde{m}}^{l}}\left(\frac{1}{2\pi r^{l}_{m,\tilde{m}}}-j\frac{1}{\lambda}\right)e^{j 2 \pi r_{m,\tilde{m}}^{l}/\lambda}, l \in \mathcal{L},\label{deviationTransmitter}
\end{align}
where $ A_{t} $ denotes the  meta-atom area at the transmitter SIM, $ x_{m,\tilde{m}}^{l} $ is the angle between the propagation direction and the normal direction of the  transmit metasurface layer  $ (l-1) $, and $ r_{m,\tilde{m}}^{l} $, is the respective transmission distance. Hence, the effect of the transmitter SIM can be written as
\begin{align}
	\bP=\bPhi^{L}\bW^{L}\cdots\bPhi^{2}\bW^{2}\bPhi^{1}\bW^{1}\in \mathbb{C}^{M \times N_{t}},\label{TransmitterSIM}
\end{align}
where $ \bW^{l}\in \mathbb{C}^{M \times M}, l \in \mathcal{L}/\{1\} $ is the transmission coefficient matrix between the  transmit metasurface layer $ (l-1)$ and the  transmit metasurface layer $ l $, while $ \bW^{1} \in \mathbb{C}^{M \times  N_{t}} $ is the transmission coefficient matrix from the transmit antenna array to the input metasurface layer of the transmit SIM.

In the case of the transmission coefficient from  meta-atom
 $ n $ on the  receive metasurface layer $ k $ to the  meta-atom $ \tilde{n} $ on the  receive metasurface layer $ (k-1) $, it is given by
\begin{align}
		u_{\tilde{n},n}^{k}=\frac{A_{r}cos \zeta_{\tilde{n},n}^{k}}{t_{\tilde{n},n}^{k}}\left(\frac{1}{2\pi t^{k}_{\tilde{n},n}}-j\frac{1}{\lambda}\right)e^{j 2 \pi t_{\tilde{n},n}^{k}/\lambda}, k \in \mathcal{K},\label{deviationReceiver}
\end{align}
where $ A_{r} $ is the   meta-atom area in the receiver SIM, $ \zeta_{\tilde{n},n}^{k} $ is the angle between the propagation direction and the normal direction of the receive metasurface layer $ (k-1) $, and $ t_{\tilde{n},n}^{k} $ is the corresponding transmission distance. Thus, the effect of the receiver SIM is expressed by
\begin{align}
	\bZ=\bU^{1}\bPsi^{1}\bU^{2}\bPsi^{2}\cdots\bU^{K}\bPsi^{K}\in \mathbb{C}^{N_{r} \times N},
\end{align}
where $ \bU^{k}\in \mathbb{C}^{N \times N}, k \in \mathcal{K}/\{1\} $
is the transmission coefficient matrix between the  receive metasurface layer $ k $ to the  receive metasurface layer $ (k-1) $, and $ \bU^{1}\in \mathbb{C}^{N_{r} \times N}$ is the transmission coefficient matrix from the output metasurface layer of the receiver SIM to the receive antenna array.
\begin{remark}
	Practical hardware imperfections such as innate modeling errors \cite{Liu2022} may lead to deviation of the transmission coefficients between adjacent metasurface layers from those given by \eqref{deviationTransmitter} and \eqref{deviationReceiver}. In this case, calibration of these coefficients is necessary for each individual SIM. Although the calibration process is beyond the scope of this work, one solution suggests measuring the response at the receive panel after the transmission of a known signal as mentioned in \cite{An2023}. Next, update of the transmission
	coefficients could take place after applying the standard error back-propagation algorithm \cite{LeCun2015}.
	\end{remark}

Concerning the HMIMO channel between the transmitter and receiver SIMs, it is written as \cite{Hu2022}
\begin{align}
	\bG=\bR^{1/2}_{\mathrm{R}}\tilde{\bG}\bR^{1/2}_{\mathrm{T}} \in \mathbb{C}^{N\times M},\label{channel}
\end{align}
where $ \tilde{\bG}\sim \mathcal{CN}(\b0,\mathrm{PL}\Id_{N}\otimes \Id_{M})\in \mathbb{C}^{N\times M} $ denotes the independent and identically distributed (i.i.d.) Rayleigh fading channel, $ \bR_{\mathrm{T}}\in \mathbb{C}^{M\times M} $ is the spatial correlation matrix at the transmitter SIM, and $ \bR_{\mathrm{R}}\in \mathbb{C}^{N\times N} $ is the spatial correlation matrix at the receiver SIM. Note that $ \mathrm{PL} $ corresponds to the average path loss between the transmitter and receiver SIMs. In particular, in the case of isotropic scattering and far-field propagation \cite{Pizzo2020,Dai2020}, the spatial correlation matrices at the transmitter and receiver SIMs are given by \cite{Demir2022}
\begin{align}
	[\bR_{\mathrm{T}}]_{m,\tilde{m}}&=\mathrm{sinc}(2 r_{m,\tilde{m}}/\lambda),  m\in \mathcal{M}, \tilde{m}\in \mathcal{M},\\
		[\bR_{\mathrm{R}}]_{\tilde{n},n}&=\mathrm{sinc}(2 t_{\tilde{n},n}/\lambda),  \tilde{n}\in \mathcal{N}, n\in \mathcal{N},
\end{align}
respectively.

The path loss ,which attenuates the received signal is given by \cite{Rappaport2015}
\begin{align}
	\mathrm{PL}(d)=\mathrm{PL}(d_{0})+10 b \log_{10}\left(\frac{d}{d_{0}}\right)+X_{\delta},~d \ge d_{0},
\end{align}
where $ X_{\delta} $ is a Gaussian random variable with a zero mean and 
a standard deviation $ \delta $ that depends on shadow fading, $ b $ is the path loss exponent, and $ \mathrm{PL}(d_{0})=20 \log_{10}(4 \pi d_{0}/\lambda)~\mathrm{dB}$ denotes the free space path loss at the reference distance $ d_{0} $.

The received signal vector at the destination is given by
\begin{align}
	\by=\bH \bx+\bn,
\end{align}
where $ \bx \in \mathbb{C}^{N_{t} \times 1} $ is
the transmit signal vector, $ \bn \in \mathbb{C}^{N_{r} \times 1} $ is
the noise vector distributed as $ \mathcal{CN}\left(\b0, N_{0}\Id\right) $, and $ \bH \in \mathbb{C}^{N_{r} \times N_{t}} $ is the end-to-end channel that can be written as
\begin{align}
	\bH=\bZ\bG\bP,\label{EquivalentChannel}
\end{align}
We assume that $ \EE\{\bx^{\H}\bx\}\le P $, where $ P $ is the maximum average transmit power. Also, we denote $ \bQ=\EE\{\bx \bx^{\H}\} $, where $ \bQ\succeq\b0 $ is the transmit covariance matrix. Note that the transmit
power constraint can be rewritten as $ \tr(\bQ)\le P $.

\textcolor{black}{\begin{remark}
	Please allow us to elaborate by mentioning that to the best of our knowledge, this is the first paper to leverage the hybrid digital and wave-based beamforming design. Hence, we focus on the point-to-point MIMO case and assume the channels associated with different meta-atoms have been estimated by existing methods, e.g., \cite{Nadeem2023}.
\end{remark}}

\textcolor{black}{\begin{remark}
	We highlight that the SIM aims to approach fully digital systems. The network performance and time-varying channel on point-to-point MIMO under the fully digital architecture have been well studied. Extending the proposed method to multicell networks relying on existing works, e.g., \cite{ElSawy2017} is straightforward and would motivate future research. For example, in \cite{An2023c,An2023b} the authors have evaluated the performance of SIM under multiuser scenarios, demonstrating the capability of suppressing interference by leveraging the wave-based beamforming.
\end{remark}}

\textcolor{black}{\begin{remark}
	Furthermore, the time-varying channel condition would result in the wave-based beamforming design under imperfect or outdated CSI \cite{Papazafeiropoulos2021}. This requires the robust beamforming design by extending the proposed optimization method, which is beyond the scope of this paper and left for our future research.
\end{remark}}

\subsection{Problem Formulation}
In this work, we aim at maximizing the achievable rate of the SIM-assisted HMIMO wireless communication system. For a given covariance matrix $ \bQ $, assuming Gaussian signaling, the achievable rate can be written as
\begin{align}
	R=\log_{2}\det \left(\Id+\frac{1}{N_{0}}\bH\bQ\bH^{\H}\right) (\mathrm{bit/s/Hz}),
\end{align}
where $ \bH $ is perfectly known at both the transmitter and the receiver, and depends on $ \bphi^{l},l \in \mathcal{L} $ and $ \bpsi^{k},k \in \mathcal{K} $.

Mathematically, the optimization problem can be formulated as
\begin{subequations}\label{eq:subeqns}
	\begin{align}
		(\mathcal{P})~~&\max_{\bQ,\bphi_{l},\bpsi_{k}} 	\;	f(\bQ,\bphi_{l},\bpsi_{k})=\ln\det \left(\Id+\bar{\bH}\bQ\bar{\bH}^{\H}\right)\label{Maximization1} \\
		&~	\mathrm{s.t}~~~\;\!\tr(\bQ)\le P;\bQ\succeq\b0,	\label{Maximization2} \\
		&\;\quad\;\;\;\;\;\!\!~\!	\bP=\bPhi^{L}\bW^{L}\cdots\bPhi^{2}\bW^{2}\bPhi^{1}\bW^{1},
		\label{Maximization3} \\
		&\;\quad\;\;\;\;\;\!\!~\!		\bZ=\bU^{1}\bPsi^{1}\bU^{2}\bPsi^{2}\cdots\bU^{K}\bPsi^{K},
		\label{Maximization4} \\
		&\;\quad\;\;\;\;\;\!\!~\!		\bPhi^{l}=\diag(\phi^{l}_{1}, \dots, \phi^{l}_{M}), l \in \mathcal{L},
		\label{Maximization5} \\
			&\;\quad\;\;\;\;\;\!\!~\!		\bPsi^{k}=\diag(\psi^{k}_{1}, \dots, \psi^{k}_{N}), k \in \mathcal{K},
		\label{Maximization6} \\
&\;\quad\;\;\;\;\;\!\!~\!		|	\phi^{l}_{m}|=1, m \in \mathcal{M}, l \in \mathcal{L},	\label{Maximization7} \\
	&	\;\quad\;\;\;\;\;\!\!~\!		|\psi^{k}_{n}|=1, n \in \mathcal{N}, k \in \mathcal{K}	\label{Maximization8},
	\end{align}
\end{subequations}
where we have denoted $ \bar{\bH}=\bH/ \sqrt{N_{0}} $.
\section{Achievable Rate Optimization}\label{RateOptimization}
We observe that problem $ (\mathcal{P}) $ is nonconvex with an objective function being neither concave nor convex with respect to its variables and with non-convex constant modulus constraints. Hence, contrary to conventional MIMO systems, the
water-filling algorithm cannot be used to obtain the maximum achievable rate. Also, previous proposed optimization methods on RIS-assisted systems have relied on the alternating optimization (AO) \cite{Wu2019,Zhang2020a}, where the covariance matrix and the RIS phase shifts are optimized separately in an alternating fashion. However, despite the easy implementation of the AO method, its convergence may require many iterations, which increase with the number of RIS elements \cite{Perovic2021}. Given that SIM-assisted HMIMO systems suggest a case, where each metasurface has a large number of elements, AO is not recommended. These observations motivate us to propose to apply an efficient projected gradient method similar to \cite{Li2015b,Perovic2021}, where the covariance matrix and the phase shifts at the transmitter and receiver SIMs are optimized simultaneously.

\subsection{Proposed Algorithm}
According to the proposed approach, we perform a simultaneous optimization of all variables in each iteration instead of optimizing them a single variable at a time. Section \ref{Numerical} will demonstrate the faster convergence compared to the AO method.

 We outline the proposed algorithm solving \eqref{eq:subeqns} in Algorithm \ref{Algoa1}. The central concept assumes to start from an arbitrary point $ (\bQ^{0},\bphi_{l}^{0},\bpsi_{l}^{0}) $, and move towards $ \nabla f(\bQ,\bphi_{l},\bpsi_{k}) $, i.e., the gradient of $ f(\bQ,\bphi_{l},\bpsi_{k}) $. The parameter $ \mu_{n}^{q} >0$ for $ q=1,2,3 $ determines the step of this move.
 
 For the description of the proposed algorithm, we make use of the following sets.
 \begin{align}
 	\mathcal{Q}&=\{\bQ \in \mathbb{C}^{}:\tr(\bQ)\le P;\bQ\succeq\b0 \},\\
 		\Phi_{l}&=\{\bphi_{l}\in \mathbb{C}^{M \times 1}: |\phi^{l}_{i}|=1, i=1,\ldots, M\},\\
 		\Psi_{k}&=\{\bpsi_{k}\in \mathbb{C}^{N \times 1}: |\psi^{k}_{i}|=1, i=1,\ldots, N\}.
 \end{align}

Note that before each step towards the gradient of $ f(\bQ,\bphi_{l},\bpsi_{k}) $, the newly computed points $ \bQ,\bphi_{l},\bpsi_{k} $ are projected onto their feasible sets $ \mathcal{Q} $, $ \Phi_{l} $, and $ \Psi_{k} $, respectively. Otherwise, the ensuing updated point may be found outside of the feasible set. Below, we provide $ \nabla_{\bQ}f(\bQ,\bphi_{l},\bpsi_{k}) $, $\nabla_{\bphi_{l}}f(\bQ,\bphi_{l},\bpsi_{k}) $, and $ \nabla_{\bpsi_{k}}f(\bQ,\bphi_{l},\bpsi_{k}) $, which correspond to the directions where the rate of change of $ f(\bQ,\bphi_{l},\bpsi_{k}) $ becomes maximum \cite[Theorem3.4]{hjorungnes:2011}.

\begin{algorithm}[th]
	\caption{Projected Gradient Ascent Method for SIM-assisted HMIMO Systems \label{Algoa1}}
	\begin{algorithmic}[1]
		\STATE Input: $\bQ^{0},\bphi_{l}^{0},\bpsi_{k}^{0},\mu_{n}^{\textcolor{black}{q}}>0$ \textcolor{black}{for $ q=1,2,3 $}.
		\STATE \textbf{for} $ n=1,2,\ldots \textbf{do} $
		\STATE ~~~~~$\bQ^{n+1}=P_{Q}(\bQ^{n}+\mu_{n}^{\textcolor{black}{1}}\nabla_{\mathcal{Q}}f(\bQ^{n},\bphi_{l}^{n},\bpsi_{k}^{n}))$
		\STATE ~~~~~$\bphi_{l}^{n+1}=P_{\Phi_{l}}(\bphi_{l}^{n}+\mu_{n}^{\textcolor{black}{2}}\nabla_{\bphi_{l}}f(\bQ^{n},\bphi_{l}^{n},\bpsi_{k}^{n}))$
		\STATE ~~~~~$\bpsi_{k}^{n+1}=P_{\Psi_{k}}(\bpsi_{k}^{n}+\mu_{n}^{\textcolor{black}{3}}\nabla_{\bpsi_{k}}f(\bQ^{n},\bphi_{l}^{n},\bpsi_{k}^{n}))$
				\STATE \textbf{end for}
	\end{algorithmic}
\end{algorithm} 	

Note that $ P_{Q}(\cdot) $, $ P_{\Phi_{l}}(\cdot) $, and $ P_{\Psi_{k}}(\cdot) $ denote the projections onto $ \mathcal{Q} $, $ \Phi_{l} $, and $ \Psi_{k} $, respectively.

\subsection{Complex-Valued Gradients of $f(\bQ,\bphi_{l},\bpsi_{k}) $}
In this subsection, we provide the gradients of $f(\bQ,\bphi_{l},\bpsi_{k}) $.
\begin{lemma}\label{lemmaGradient}
	The gradients of $f(\bQ,\bphi_{l},\bpsi_{k}) $ with respect to $ \bQ^{*}$, $\bphi_{l}^{*}$, and $\bpsi_{k}^{*}$  are given by
	\begin{align}
		\nabla_{\bQ}f(\bQ,\bphi_{l},\bpsi_{k})&=\bar{\bH}^{\H}	\bK(\bQ,\bphi_{l},\bpsi_{k})\bar{\bH},\label{gradient1}\\
		\nabla_{\bphi_{l}}f(\bQ,\bphi_{l},\bpsi_{k})&=\diag( \bA_{l}^{\H}),\label{gradient2}\\
		\nabla_{\bpsi_{k}}f(\bQ,\bphi_{l},\bpsi_{k})&=\diag( \bC_{k}^{\H}),\label{gradient3}
	\end{align}
	where 
	\begin{align}
		&\bK(\bQ,\bphi_{l},\bpsi_{k})=\left(\Id+\bar{\bH}\bQ\bar{\bH}^{\H}\right)^{-1},\\
		&\bA_{l}=\bW^{l}\bPhi^{l-1}\bW^{l-1}\cdots \bPhi^{1}\bW^{1}\bQ\bar{\bH}^{\H}\bK \bZ\bar{\bG}\bPhi^{L}\nn\\
		&\times\bW^{L}\cdots\bPhi^{l+1}\bW^{l+1},\\
				&	\bC_{k}=\bU^{k}\bPsi^{k-1}\bU^{k-1}\cdots \bPsi^{1}\bU^{1}\bK\bar{\bH}\bQ\bP^{\H} \bar{\bG}^{\H}\bPsi^{K} \nn\\
				&\times\bU^{K}\cdots \bPsi^{k+1}\bU^{k+1}.
	\end{align}
\end{lemma}
\begin{proof}
	Please see Appendix~\ref{lem1}.	
\end{proof}

\subsection{Projection Operations of Algorithm \ref{Algoa1}}
Regarding the constraint $ |	\phi^{l}_{m}|=1 $, it indicates that $ \phi^{l}_{m} $ should be located on the unit circle in the complex plane. In particular, let a given point $ \bu_{l} \in \mathbb{C}^{M \times 1} $, then for the vector $ \bar{\bu}_{l} $ of $ P_{\Phi_{l}}(\bu_{l}) $, it holds that
	\begin{align}
	\bar{u}_{l,m}=\left\{
	\begin{array}{ll}
		\frac{u_{l,m}}{|u_{l,m}|} & u_{l,m}\ne 0 \\
		e^{j \phi^{l}_{m}}, \phi^{l}_{m} \in [0, 2 \pi] &u_{l,m}=0 \\
	\end{array}, m=1, \ldots,M.
	\right.
\end{align}

A similar observation holds for $ 	|\psi^{k}_{n}|=1 $. In the case of the projection onto $ \mathcal{Q} $, it has been presented previously, e.g., in \cite{Pham2018,Perovic2021}. Herein, we present it briefly. Specifically, we have that the projection of $ \bY \succeq\b0 $ is the solution to the following optimization.
\begin{subequations}\label{eq:subeqns1}
\begin{align}
	&\min_{\bQ} \|\bQ-\bY\|^{2}	\label{Min1} \\
	&~	\mathrm{s.t}~~~\;\!\tr(\bQ)\le P;\bQ\succeq\b0.	\label{Min2}
\end{align}
\end{subequations}
By applying the eigenvalue decomposition to $ \bY $ and $ \bQ $, we have $ \bY=\bU \bSigma \bU^{\H} $ and $ \bY=\bU \bD \bU^{\H} $, where
$ \bSigma=\diag(\sigma_{1}, \ldots, \sigma_{S}) $ and $ \bD=\diag(d_{1}, \ldots, d_{S}) $, while $ \bU $ is a matrix that includes the eigenvectors. On this ground, \eqref{eq:subeqns1} can be written equivalently as
\begin{subequations}\label{eq:subeqns1}
	\begin{align}
		&\min_{d_{i}} \sum_{i=1}^{S} (d_{i}-\sigma_{i})^{2}	\label{Min3} \\
		&~	\mathrm{s.t}~~~\;\!\sum_{i=1}^{N}d_{i}\le P, d_{i}\ge 0,	\label{Min4}
	\end{align}
\end{subequations}
which can be solved by the water-filling algorithm with solution provided by
\begin{align}
	d_{i}=(\sigma_{i}-\gamma)_{+}, i=1, \ldots, M,\label{waterfilling}
\end{align}
where $ \gamma \ge 0$ is the water level.
\section{Convergence and Complexity Analyses}\label{convergence}
\subsection{Convergence Analysis}
Based on \cite{Li2015b}, in this subsection, we present the proof of the convergence of Algorithm \ref{Algoa1}, which solves Problem $ 	(\mathcal{P}) $. The proof requires to show that $ f(\bQ,\bphi_{l},\bpsi_{k}) $ has a Lipschitz continuous
gradient, where the  Lipschitz constant is $ \Lambda $ (see definition below). Next, we will argue that the convergence of  Algorithm \ref{Algoa1} holds if the step size obeys to \textcolor{black}{$ \mu_{n}^{q} \le \frac{1}{\Lambda} $ for $ q=1,2,3 $}. 

\begin{definition}
	A function $ f(\bx) $ is $ \Lambda $-Lipschitz continuous, or else $ \Lambda $-smooth over a set $ \mathcal{X} $, if for all $ \bx,\by \in \mathcal{X} $, we have
	\begin{align}
		\|\nabla f(\by)-\nabla f(\bx)\|\le \Lambda \|\by-\bx\|.
	\end{align}
\end{definition}

According to this definition, in our case, we have to show \eqref{gradient10} at the bottom of the next page.
	\begin{figure*}
	\begin{align}
	&\Big(\|  \nabla_{\bQ}f(\bQ^{1},\bphi_{l}^{1},\bpsi_{k}^{1})-  \nabla_{\bQ}f(\bQ^{2},\bphi_{l}^{2},\bpsi_{k}^{2})\|^{2} +\|  \nabla_{\bphi_{l}}f(\bQ^{1},\bphi_{l}^{1},\bpsi_{k}^{1})-  \nabla_{\bphi_{l}}f(\bQ^{2},\bphi_{l}^{2},\bpsi_{k}^{2})\|^{2}\nn\\&+\|  \nabla_{\bpsi_{k}}f(\bQ^{1},\bphi_{l}^{1},\bpsi_{k}^{1})-  \nabla_{\bpsi_{k}}f(\bQ^{2},\bphi_{l}^{2},\bpsi_{k}^{2})\|^{2}	\Big)^{1/2}\le \Lambda	\Big(\|\bQ^{1}-\bQ^{2}\|^{2}+\|\bphi_{l}^{1}-\bphi_{l}^{2}\|+\|\bpsi_{k}^{1}-\bpsi_{k}^{2}\|^{2}\Big)^{1/2}.\label{gradient10}
	\end{align}
	\hrulefill
\end{figure*}

\begin{proposition}\label{proposition1}
	The gradients $ 	\nabla_{\bQ}f(\bQ,\bphi_{l},\bpsi_{k}) $, $ \nabla_{\bphi_{l}}f(\bQ,\bphi_{l},\bpsi_{k}) $, and $	\nabla_{\bpsi_{k}}f(\bQ,\bphi_{l},\bpsi_{k}) $ obey to the following inequalities

	\begin{align}
	&	\|  \nabla_{\bQ}f(\bQ^{1},\bphi_{l}^{1},\bpsi_{k}^{1})-  \nabla_{\bQ}f(\bQ^{2},\bphi_{l}^{2},\bpsi_{k}^{2})\|\nn\\
		&	\le bf \big( \frac{2 b d^{2}}{c^{2}} \|\bP^{1}- \bP^{2}\|+ \frac{2 bf d}{c} \big(\|\bZ^{1}-\bZ^{2}\|\big),\\
		&\|  \nabla_{\bphi_{l}}f(\bQ^{1},\bphi_{l}^{1},\bpsi_{k}^{1})-  \nabla_{\bphi_{l}}f(\bQ^{2},\bphi_{l}^{2},\bpsi_{k}^{2})\|
		\nn\\
		&\le a_{l} b c |\bQ^{1}-\bQ^{2}\|+a_{l} b P\left(c+ \frac{bdf}{c N_{0}}\right)\|\bZ^{1}-\bZ^{2}\|\nn
		\\&+a_{l} P\frac{b^{2}d^{2}}{c^{2}N_{0}} \|\bP^{1}- \bP^{2}\|,\\
	&	\|		\nabla_{\bpsi_{k}}f(\bQ^{1},\bphi_{l}^{1},\bpsi_{k}^{1}) -\nabla_{\bpsi_{k}}f(\bQ^{2},\bphi_{l}^{2},\bpsi_{k}^{2})\|\nn\\&\le b_{k} \frac{b d}{c} |\bQ^{1}-\bQ^{2}\|+
					b_{k} b P \left(c+ \frac{bdf}{c N_{0}}\right)\|\bZ^{1}-\bZ^{2}\|\nn
		\\&+b_{k} P\frac{b^{2}d^{2}}{c^{2} N_{0}} \|\bP^{1}- \bP^{2}\|,
	\end{align}
where
\begin{align}
	a_{l}&=	\lambda_{\mathrm{max}}(\bW^{l}\bPhi^{l-1}\bW^{l-1}\cdots \bPhi^{1}\bW^{1})\nn\\
	&\times	\lambda_{\mathrm{max}}(\bPhi^{L}\bW^{L}\cdots\bPhi^{l+1}\bW^{l+1}),\label{al}\\
	b_{k}&=	\lambda_{\mathrm{max}}(\bU^{k}\bPsi^{k-1}\bU^{k-1}\cdots \bPsi^{1}\bU^{1})\nn\\
	&\times	\lambda_{\mathrm{max}}(\bPsi^{K} \bU^{K}\cdots \bPsi^{k+1}\bU^{k+1}),\label{bk}\\
	b&=\lambda_{\mathrm{max}}(\bG)\label{bg},\\
	c&=\lambda_{\mathrm{max}}(\bar{\bH})\label{Eqc},\\
	d&=c\lambda_{\mathrm{max}}(\bZ)\label{Eqd},\\
	f&=\lambda_{\mathrm{max}}(\bP)\label{Eqf}.
\end{align}
\end{proposition}
\begin{proof}
Please see Appendix~\ref{Prop1}.	
\end{proof}

\begin{theorem}\label{Theorem1}
The objective $ f(\bQ,\bphi_{l},\bpsi_{k}) $ is $ \Lambda $-Lipschitz continuous with a parameter $ \Lambda $ provided by
\begin{align}
	\Lambda=\sqrt{\max(\Lambda_{\bQ}^{2},\Lambda_{\bphi_{l}}^{2},\Lambda_{\bpsi_{k}}^{2})},\label{theo1Eq}
\end{align}
where 
\begin{align}
	\Lambda_{\bQ}^{2}&=(a_{l} b c + b_{k} \frac{b d}{c})^{2},\\
	\Lambda_{\bphi_{l}}^{2}&= ( \frac{2 b^{2}f d^{2}}{c^{2}}+(a_{l}+b_{k}) P\frac{b^{2}d^{2}}{c^{2}N_{0}} )^{2},\\
	\Lambda_{\bpsi_{k}}^{2}&=(\frac{2 b^{2}f^{2} d}{c}+(a_{l}+b_{k}) P\frac{b^{2}d^{2}}{c^{2} N_{0}})^{2}.
\end{align}
\end{theorem}

\begin{proof}
	Please see Appendix~\ref{Theo1}.	
\end{proof}

The following theorem indicates the convergence of Algorithm \ref{Algoa1}.
\begin{theorem}\label{Theorem2}
The point $ \bQ^{*},\bphi_{l}^{*},\bpsi_{k}^{*} $ corresponds to a critical point of Problem $ 	(\mathcal{P}) $, and Algorithm \ref{Algoa1} is bounded, if the step size obeys to \textcolor{black}{$ \mu_{n}^{q} < \frac{1}{\Lambda}$ for $ q=1,2,3 $} with $ \Lambda $ given by \eqref{theo1Eq}.
\end{theorem}
\begin{proof}
	Please see Appendix~\ref{Theo2}.	
\end{proof}

\textcolor{black}{Usually, Theorem 1 provides a  Lipschitz constant, which is much larger
 than the best Lipschitz constant for the gradient of the objective. As a results, the step sizes obtained by Theorem 2 are very small, which lead to slow convergence. For this reason, we use  a backtracking line search to obtain a  larger step size at each iteration that can speed up the convergence. In this direction,  a line search procedure is adopted according to the Armijo–Goldstein condition, which is numerically efficient as shown in Sec. V.}

\textcolor{black}{On this ground, we replace the step sizes $ \mu_{n}^q $ of Algorithm 1 by $ L_{o}^{q}\rho^{\kappa_{n}^{q}} $ for $ q=1,2,3 $, where $ L_{o}^{q} >0$, $ \rho \in (0,1) $ and $ \kappa_{n}^{q} $ is the smallest
nonnegative integer that satisfies \eqref{52below}.
\begin{align}
	&\bQ^{n+1}=P_{Q}(\bQ^{n}+\mu_{n}^{\textcolor{black}{1}}\nabla_{\mathcal{Q}}f(\bQ^{n},\bphi_{l}^{n},\bpsi_{k}^{n})),\\
&\bphi_{l}^{n+1}=P_{\Phi_{l}}(\bphi_{l}^{n}+\mu_{n}^{\textcolor{black}{2}}\nabla_{\bphi_{l}}f(\bQ^{n},\bphi_{l}^{n},\bpsi_{k}^{n})),\\
&\bpsi_{k}^{n+1}=P_{\Psi_{k}}(\bpsi_{k}^{n}+\mu_{n}^{\textcolor{black}{3}}\nabla_{\bpsi_{k}}f(\bQ^{n},\bphi_{l}^{n},\bpsi_{k}^{n})),\\
	&f(\bQ^{n+1},\bphi_{l}^{n+1},\bpsi_{k}^{n+1})\! \ge\! f(\bQ^{n},\bphi_{l}^{n},\bpsi_{k}^{n})+\delta^{1} \|\bQ^{n+1}-\bQ^{n}\|^{2}\nn\\
	&+\delta^{2}\|\bphi_{l}^{n+1}-\bphi_{l}^{n}\|^{2}+\delta^{3}\|\bpsi_{k}^{n+1}-\bpsi_{k}^{n}\|^{2},\label{52below}
\end{align}
where $ \delta^q>0 $. The backtracking line search can be obtained through
an iterative procedure, which is guaranteed to terminate after
a finite number of iterations since $ f $ is  $ \Lambda $-smooth.}

\subsection{Complexity Analysis}
Herein, we present the computational complexity of Algorithm \eqref{Algoa1} by providing the complex multiplications per iteration.

Steps $ 3-5 $ of Algorithm \ref{Algoa1} determine the complexity of the proposed approach. In particular, the computation of $ 	\bH=\bZ\bG\bP $ requires $ N_{r}M(N+N_{t})$ multiplications. For the computation of $ \nabla_{\mathcal{Q}}f(\bQ,\bphi_{l},\bpsi_{k})) $, we need to compute the matrix inversion of $ \bK(\bQ,\bphi_{l},\bpsi_{k})=\left(\Id+\bar{\bH}\bQ\bar{\bH}^{\H}\right)^{-1} $. Instead of this, we focus on $\bV= \bK(\bQ,\bphi_{l},\bpsi_{k})\bar{\bH} \in \mathbb{C}^{N_{t}\times N_{t}} $ as provided in \eqref{gradient1}. We can simplify its computation by observing that $ \bV $ is the result to the linear system $ \left(\Id+\bar{\bH}\bQ \bar{\bH}^{\H}\right)\bX=\bar{\bH} $. 
For $ \bar{\bH}\bQ \bar{\bH}^{\H} $, we need $ N_{t}^{2}N_{r} $ multiplications to form $ \bar{\bH}\bQ $, and then $ (N_{r}^{2}+N_{t})N_{t}/2 $ multiplications for the multiplication between $ \bar{\bH}\bQ $ and $ \bar{\bH}^{\H} $. A Cholesky decomposition leads to the solution of the linear system with a complexity equal to $ \mathcal{O}(N_{r}^3+N_{r}^{2}N_{t})$. Overall, $ \bV $ requires $ \mathcal{O}(N_{r}^3+N_{r}^{2}N_{t}+\frac{3}{2}N_{t} N_{r}^{2})$. Hence, $ \bar{\bH}^{\H} \bV $ requires $ N_{t}^{2}N_{r} $ multiplications.

For the computation of $ \nabla_{\bphi_{l}}f(\bQ,\bphi_{l},\bpsi_{k}) $, we need to compute $ \bA_{l} $, which requires $ M^{2}N_{t}+MN_{t}^{2}+N_{t}^{2}N_{r}+N_{r}^{2}N+M N N_{r} +M^{2} N+(L-2)M^{3}$ multiplications. Similarly, the  computation of $ \nabla_{\bpsi_{k}}f(\bQ,\bphi_{l},\bpsi_{k}) $, we need to compute $ \bC_{K} $, which requires $NN_{r}^{2}+N_{t}N_{r}^{2}+N_{t}^{2}N_{r}+N_{t}^{2}M+M N N_{t} +M N^{2}+(K-2)N^{3}$ multiplications. Thus the computation of all gradients requires $ NN_{r}^{2}+N_{t}N_{r}^{2}+ M^{2}N_{t}+MN_{t}^{2}+2N_{t}^{2}N_{r}+N_{t}^{2}M+N_{r}^{2}N+M N(N^{t}+ N_{r}) +M N (M+N)+(L-2)M^{3} +(K-2)N^{3}$ multiplications.

In the case of the multiplication of $ \mu_{n}^{1} $ with $ \nabla_{\mathcal{Q}}f(\bQ^{n},\bphi_{l}^{n},\bpsi_{k}^{n}) $, we need $ N_{t}^{2}/2 $ operations, while the projection of $ \bQ^{n}+\mu_{n}^{1}\nabla_{\mathcal{Q}}f(\bQ^{n},\bphi_{l}^{n},\bpsi_{k}^{n}) $ onto $ \mathcal{Q} $, we need $ \mathcal{O}(N_{t}^{2}) $ operations for the water-filling algorithm in \eqref{waterfilling}, $ N_{t}^{2}+(N_{t}^{2}+N_{t})N_{t}/2 $ operations for the multiplication $ \bQ=\bU\bD \bU^{\H} $, and $ \mathcal{O}(N_{t}^{3}) $ operations for the eigenvalue decomposition. The multiplication of $ \mu_{n}^{2} $ with $ \nabla_{\bphi_{l}}f(\bQ,\bphi_{l},\bpsi_{k}) $ and the projection of the result onto $ \bphi_{l} $ requires $ 3 M $ multiplications, while $ 3 N $ multiplications are required for the same operations regarding $ \nabla_{\bpsi_{k}}f(\bQ,\bphi_{l},\bpsi_{k}) $.

Overall, the complexity of the proposed Algorithm \ref{Algoa1} per iteration is given by
\begin{align}
&	C_{\mathrm{Alg.}1, \mathrm{IT}} =\mathcal{O}\big(N_{t}^{3}+2N_{t}^{2}+(N_{t}^{2}+N_{t})N_{t}/2+NN_{r}^{2}\nn\\
&+N_{t}N_{r}^{2}+ M^{2}N_{t}+MN_{t}^{2}+2N_{t}^{2}N_{r}+N_{t}^{2}M+N_{r}^{2}N\nn\\
&+M N(N^{t}+ N_{r}) +(M N+3) (M+N)\nn\\
&+(L-2)M^{3} +(K-2)N^{3}\big),
\end{align}
which, for large SIMS, can be approximated as $ \mathcal{O}((M N+3) (M+N)+(L-2)M^{3} +(K-2)N^{3} )$.

\section{Numerical Results}\label{Numerical}
In this section, we present numerical results for the achievable rate in a SIM-assisted HMIMO system. With respect to the simulation setup, we assume that the thicknesses of the transmitter and receiver SIMs are $ D_{t}=D_{r}=0.04~\mathrm{m} $, while the corresponding transmission coefficients are given by \eqref{deviationTransmitter} and \eqref{deviationReceiver}, respectively. Also, the HMIMO channel is generated according to \eqref{channel}. The operation frequency is at $ f_{0} = 6~\mathrm{GHz} $, which means that the 	wavelength $ \lambda = 50~\mathrm{mm} $. Regarding the large-scale fading, we account for a reference distance of $ d_{0} = 1~\mathrm{m} $, assign $ b = 3.5 $ and $ \delta = 9~\mathrm{dB} $, while the distance between the transmitter and the receiver is $ 	d = 250~\mathrm{m} $. Also, the total transmit power is $ P = 20~\mathrm{dBm}$, and the variance of the receiver noise is $N_{0}=110~\mathrm{dBm} $. Moreover, we assume \textcolor{black}{$ N_{t} =N_{t}=10$}, $ M=N=100 $,  $ L=K=7 $, $ t_{e,r}=r_{e,t}=\lambda/2 $, unless otherwise specified. In this way, we obtain $ A_{t} $ and $ A_{r} $. \textcolor{black}{The line
search procedure for the proposed gradient algorithms uses $ L_{0}^{q} = 10^{4} $, $ \delta^{q} = 10^{-5} $, and $ \rho = 1/2 $. Moreover, the
minimum allowed step size value is the largest step size value lower than $ 10^{-4} $. Generally, we use equal step sizes, while Fig. 6 will depict the impact of different step sizes.}

\begin{figure}%
	\centering
	\subfigure[Achievable rate of the proposed SIM-assisted HMIMO  versus the number of iterations for various sets of numbers of metasurfaces and meta-atoms per metasurface.]{	\includegraphics[width=0.9\linewidth]{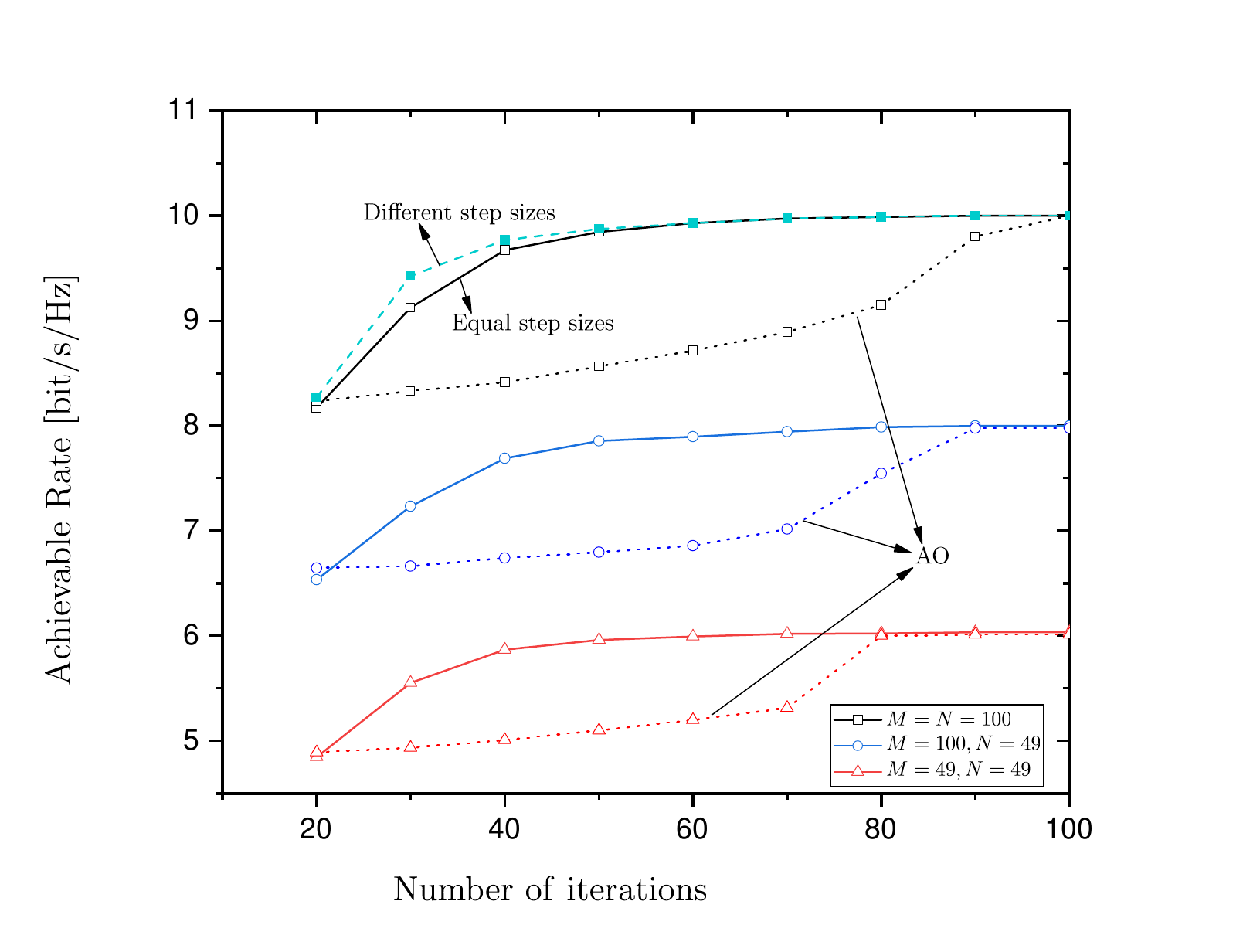}}\qquad
	\subfigure[Achievable rate of the proposed SIM-assisted HMIMO  versus the number of channel realizations for various sets of numbers  meta-atoms per metasurface.]{	\includegraphics[width=0.9\linewidth]{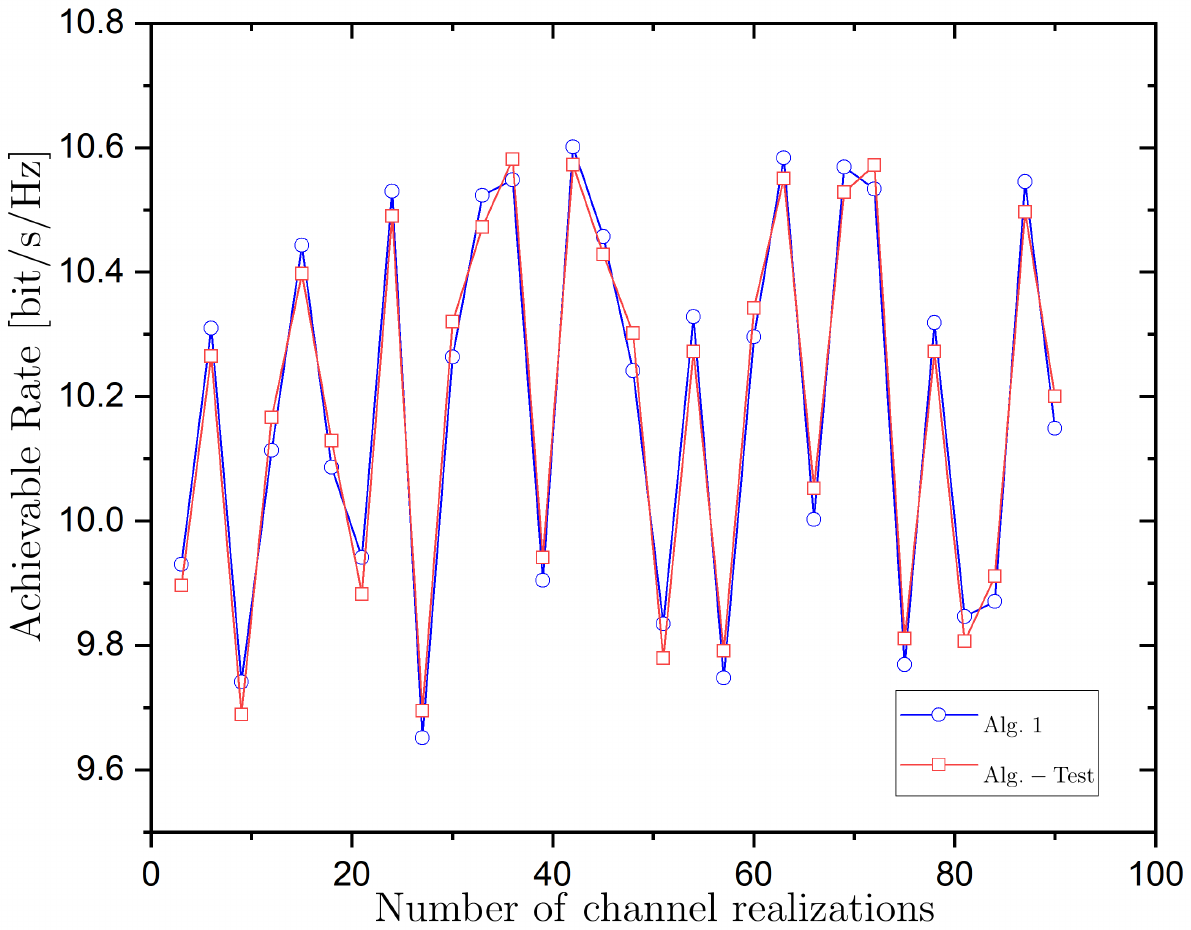}}\\
	\caption{\textcolor{black}{Convergence of the proposed algorithm and dependence on the initialization.}}
	\label{fig2}
\end{figure}

 In Fig. \ref{fig2}.(a), we depict the convergence of the proposed algorithm, i.e., Algorithm \ref{Algoa1}. Specifically, we have drawn the achievable rate of the proposed SIM-assisted HMIMO  versus the number of iterations for various sets of  meta-atoms per metasurface. Notably, we observe the fast convergence of the algorithm in all cases. For example, when $M=N=100$, the algorithm converges in $ 60 $ iterations. Furthermore, we notice that by increasing the number of  meta-atoms, more iterations are required to reach convergence. The reason behind this observation is that the amount of optimization variables increases and the relevant search space is enlarged. Apart from this, we note that these increases, i.e., in terms of the number of meta-atoms lead to higher complexity of each iteration of Algorithm \ref{Algoa1} as mentioned in Sec. \ref{RateOptimization}. Furthermore, it is shown that the gradient-based
optimization method converges quickly to the optimum rate, while the AO method requires many iterations.  Note that as the parameter values $ M,N $ increase, the AO method approaches the optimum rate slower.

In addition, the non-convexity of the optimization problem means that its result depends on the initial point. In other words, different initial points lead to different locally optimal solutions. Fig. \ref{fig2}.(b) explores this dependence on the initialization by executing the algorithm for $ 30 $ channel realizations. According to its description, the initialization of Algorithm 1 assumes that $\bQ^{0}=\Id_{S}, \bphi_{l}^{0}=\exp\left(j\pi/2\right)\one_{M}$, $ \bpsi_{k}^{0} =\exp\left(j\pi/2\right)\one_{N} $. ``Alg. 1-Test'' in the figure considers  the best initial point out of $ 50 $ random initial points for each channel instance. The figure reveals that different initializations lead to different solutions. Also, it is shown that the achievable rate in both cases is almost the same. This means that the selection of these initial values for initialization is a good decision. \textcolor{black}{In the same figure, we have plotted the case corresponding to different step sizes during the optimization of the three variables, where $ L_{o}^{3}=0.5L_{o}^{1} $ and $ L_{o}^{2}=0.1L_{o}^{1} $ with $ L_{0}^{1} = 10^{4} $. As can be seen, in this case, the optimal value is achieved  sooner.}

\textcolor{black}{In practice, we do not wait  an optimization algorithm to reach a critical point but a close value. For this reason, in the following table, we  present the computational complexity between the proposed 	projected gradient ascent  method and the AO method to acquire an achievable rate that is equal to $ 95 \,\% $ of the average achievable rate at the $ 100 $th iteration. The complexity of the AO is characterized in terms of the number of outer iterations $ I_{\mathrm{OI}} $, which is required to obtain the optimal achievable rate. Note that one outer iteration is actually a sequence of $ M+N + 1 $ conventional iterations, where $ \bphi^{l}=\{\phi_{m}^{l}\}_{m=1}^{M} $, $ \bpsi^{l}=\{\psi_{n}^{l}\}_{n=1}^{N} $, and $ \bQ $. $ C_{\mathrm{Alg.}1, \mathrm{IT}} $ is  the computational complexity per iteration while $ I_{\mathrm{Alg.}1}  $ expresses the number of iterations required to obtain the optimal achievable rate. Their product results in the total computational complexity $ C_{\mathrm{Alg.}1} $ of the  proposed 			projected gradient ascent  method.  The computational complexity of the AO is described by $ C_{AO} .$ Table \ref{tab:Comp} presents the computational complexity 	of the proposed algorithm and the AO method for varying $ M $.  We observe that $ C_{\mathrm{Alg.}1} $ becomes
	larger with an increase of $ M $. Moreover, the computational complexity of AO increases in proportion to $ M $. However, as can be seen, the complexity in the AO case is greater than the proposed algorithm.   } 
\begin{table}[t]
\caption{\textcolor{black}{Computational complexity comparison between the proposed 			projected gradient ascent  method and the AO method to reach 95\,\% of the average achievable
			rate at the 100th iteration.\label{tab:Comp}}}
	\centering{}\textcolor{black}{}%
	\begin{tabular}{cccccc}
		\toprule 
		 \textcolor{black}{\footnotesize{}$M$} & \textcolor{black}{\footnotesize{}$I_{\mathrm{Alg.}1}$} & \textcolor{black}{\footnotesize{}$C_{\mathrm{Alg.}1, \mathrm{IT}}$} & \textcolor{black}{\footnotesize{}$C_{\mathrm{Alg.}1}$} & \textcolor{black}{\footnotesize{}$I_{\mathrm{OI}}$} & \textcolor{black}{\footnotesize{}$C_{\mathrm{AO}}$}\tabularnewline
		\midrule
		\multirow{4}{*}{}  \textcolor{black}{\footnotesize{}5} & \textcolor{black}{\footnotesize{}97} & \textcolor{black}{\footnotesize{}11357} & \textcolor{black}{\footnotesize{}1101629} & \textcolor{black}{\footnotesize{}1} & \textcolor{black}{\footnotesize{}1648395}\tabularnewline
 \,	\textcolor{black}{\footnotesize{}25}  & \textcolor{black}{\footnotesize{}86} & \textcolor{black}{\footnotesize{}20114} & \textcolor{black}{\footnotesize{}1729804} & \textcolor{black}{\footnotesize{}1} & \textcolor{black}{\footnotesize{}2793758}\tabularnewline
	\,	\textcolor{black}{\footnotesize{}60} & \textcolor{black}{\footnotesize{}71} & \textcolor{black}{\footnotesize{}32646} & \textcolor{black}{\footnotesize{}2317866} & \textcolor{black}{\footnotesize{}1} & \textcolor{black}{\footnotesize{}3594682}\tabularnewline
		\,	\textcolor{black}{\footnotesize{}100} & \textcolor{black}{\footnotesize{}62} & \textcolor{black}{\footnotesize{}51476} & \textcolor{black}{\footnotesize{}3208252} & \textcolor{black}{\footnotesize{}1} & \textcolor{black}{\footnotesize{}4732857}\tabularnewline
		\bottomrule
	\end{tabular}
\end{table}

\begin{figure}[!h]
	\begin{center}
		\includegraphics[width=0.9\linewidth]{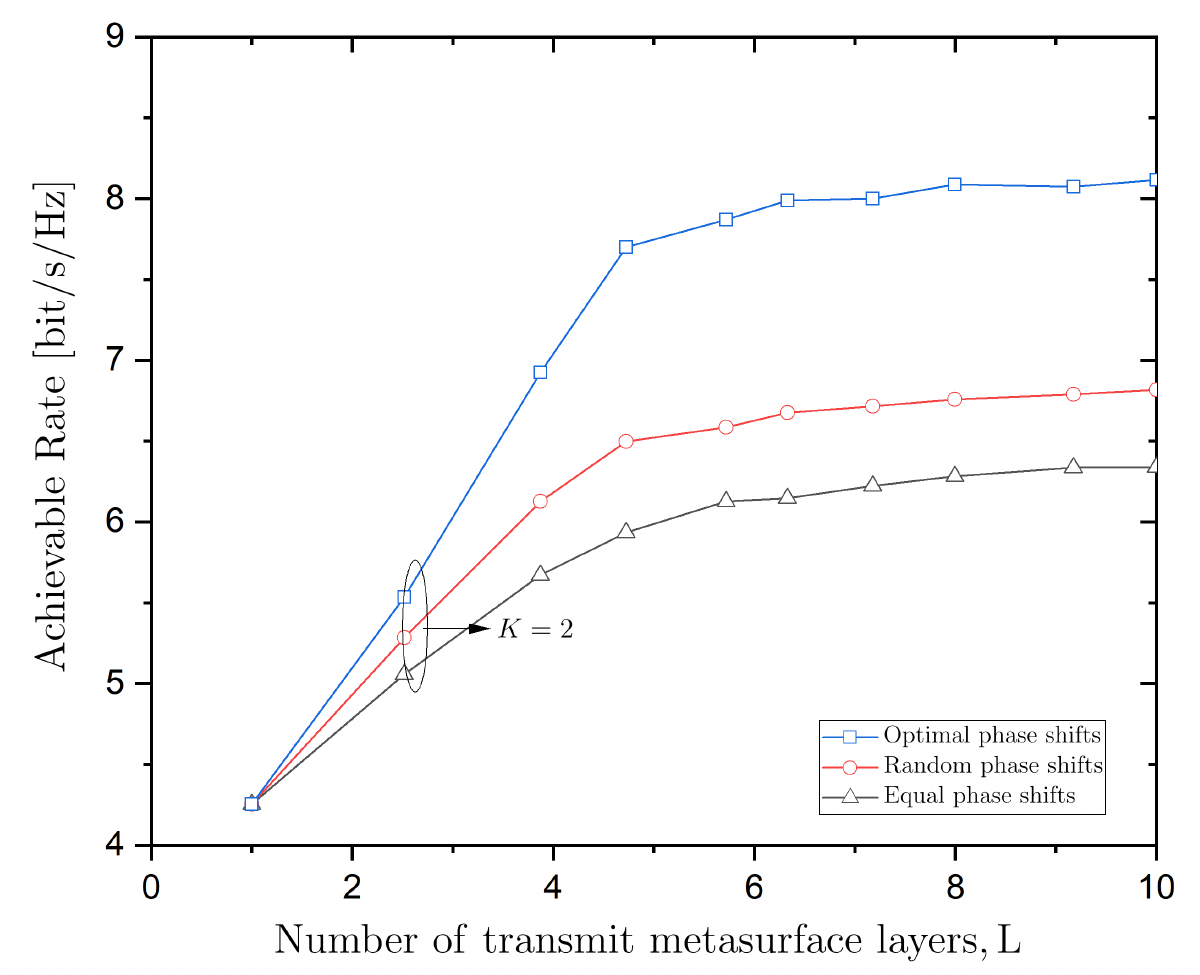}
		\caption{Achievable rate versus the number of transmit metasurfaces layers $ L $ for $ K =2$ layers of receive metasurfaces. Impact of phase shifts  optimization of the SIMs.}
		\label{fig6}
	\end{center}
\end{figure}

In Fig. \ref{fig6}, we show the achievable rate versus the number of transmit metasurfaces layers $ L $ for $ K =2$ layers of receive metasurfaces, and we focus on the impact of optimizing the phase shifts of the SIMs. The case of optimal phase shifts achieves the best performance, while the cases random phase shifts and equal phase shifts achieve lower rate. In particular, the line, corresponding to equal phase shifts, presents the worst performance.

\begin{figure}[!h]
	\begin{center}
		\includegraphics[width=0.9\linewidth]{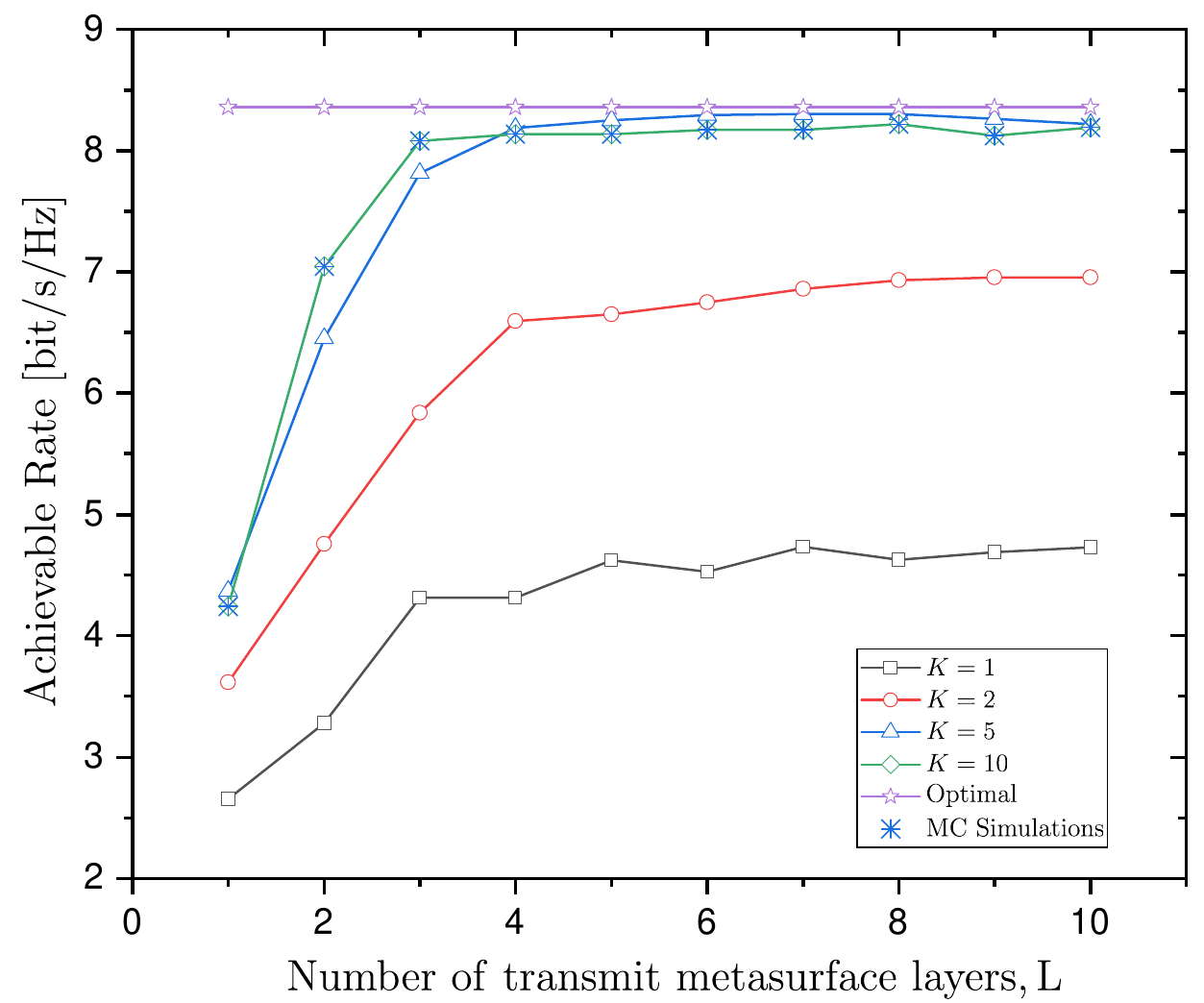}
		\caption{Achievable rate versus the number of transmit metasurfaces layers $ L $ while varying the number of receiver metasurface layers $ K $.}
		\label{fig4}
	\end{center}
\end{figure}

In Fig. \ref{fig4}, we depict the achievable rate versus the number of transmit metasurfaces layers $ L $ while varying the number of receiver metasurface layers $ K $. We observe that the achievable rate saturates as the number of transmit metasurface layers increases, i.e., $ L\ge 3 $. Also, it is shown that after a certain number of metasurfaces, the system approaches the rate obtained with only digital precoding. \textcolor{black}{Digital precoding (conventional MIMO system) has been implemented for the sake of comparison.} Of course, care should be taken regarding the thickness of the SIM since densely implemented metasurfaces may lead to performance loss. This can be met if their number is increased after a threshold. Notably, under these values, by increasing $ L $ and $ K $ after $3 $ and $ 6 $, respectively, does not improve the achievable rate.

\begin{figure}%
	\centering
	\subfigure[Achievable rate versus the number of meta-atoms per transmit layer.]{	\includegraphics[width=0.9\linewidth]{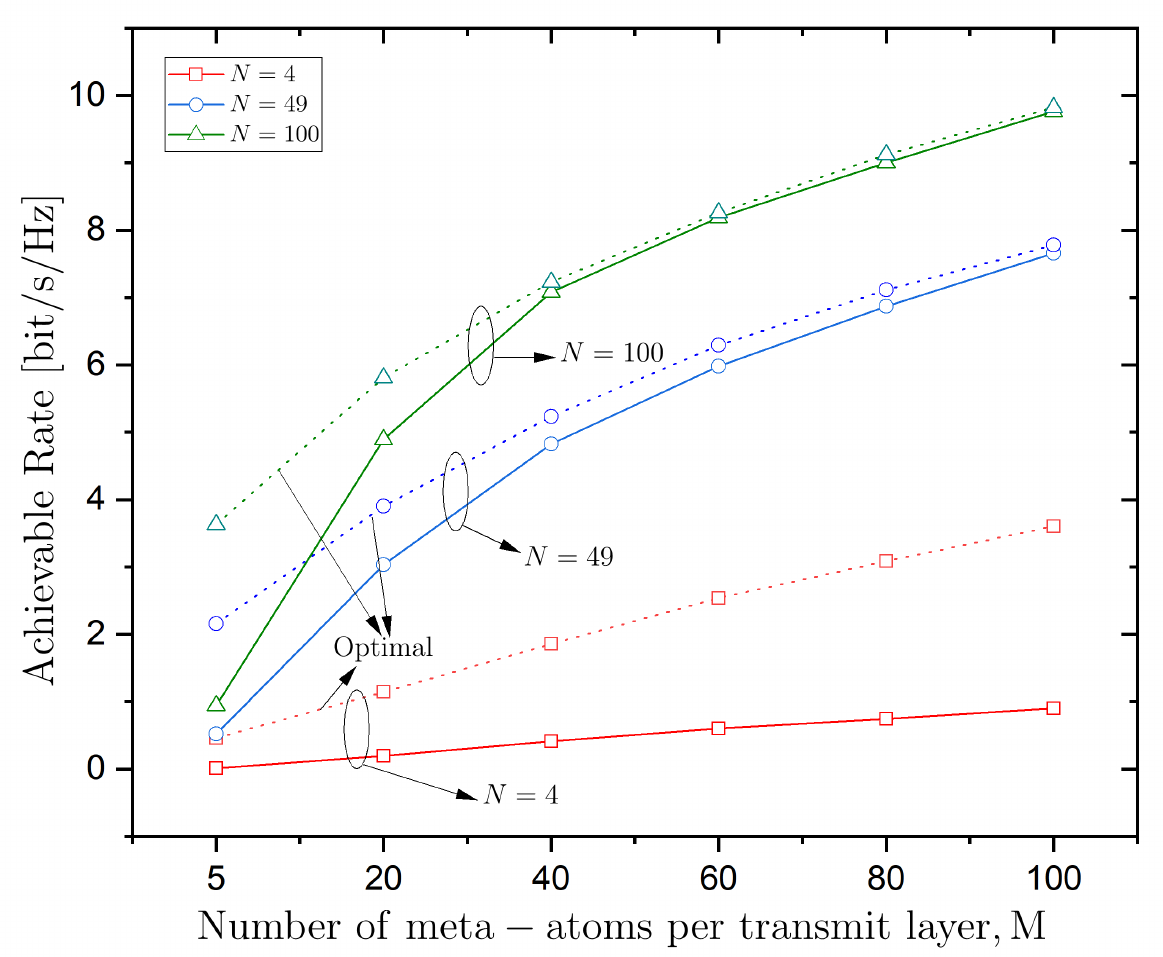}}\qquad
	\subfigure[Achievable rate versus the number of transmit antennas]{	\includegraphics[width=0.9\linewidth]{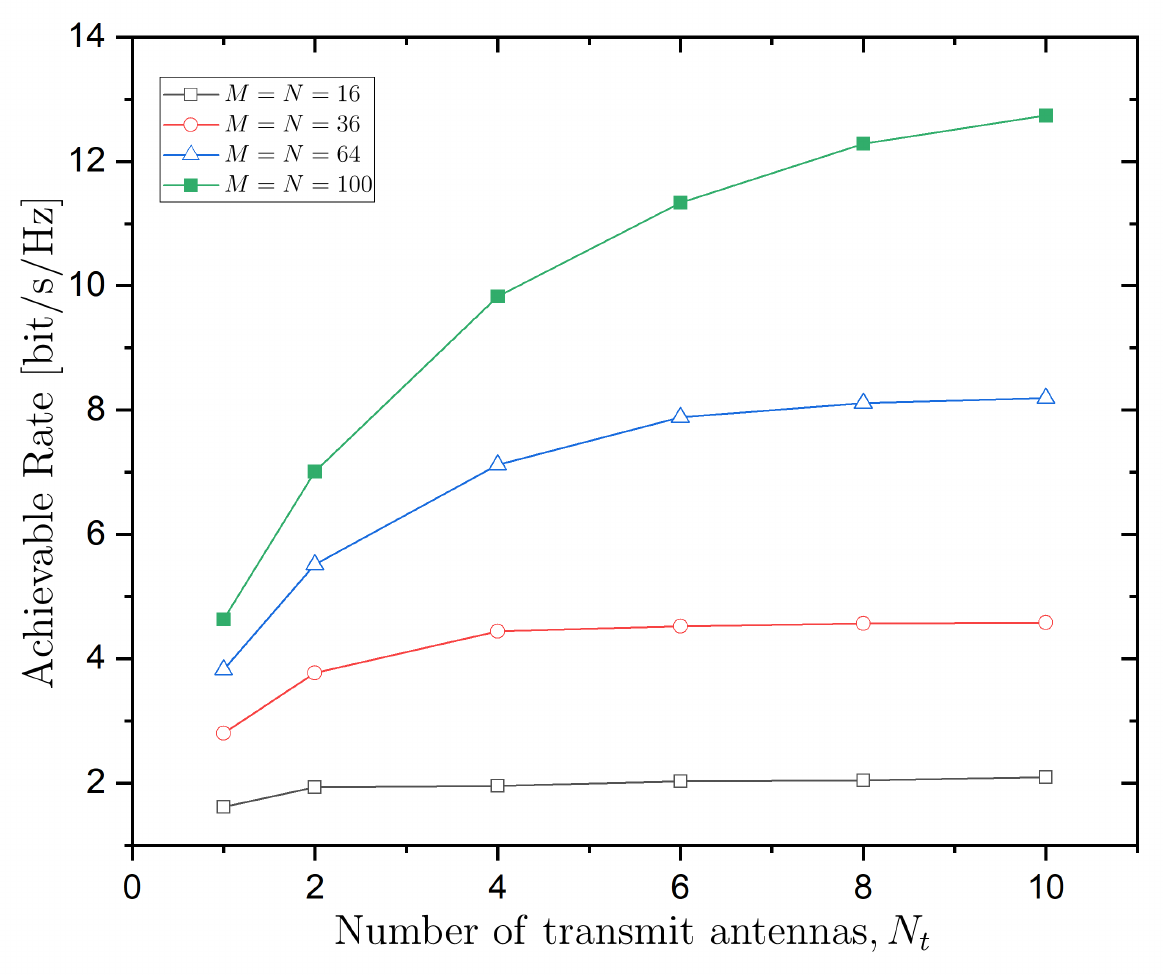}}\\
	\caption{\textcolor{black}{Impact of the numbers of meta-atoms per transmit layer and transmit antennas.}}
	\label{fig5}
\end{figure}
In Fig. \ref{fig5}.(a), we illustrate the achievable rate versus the number of meta-atoms per transmit layer. As can be seen, the rate increases monotonically when the number of meta-atoms increase at the transmitter and receiver SIMs. \textcolor{black}{Moreover, it is better to employ the receiver SIM with more meta-atoms per surface. For example,  the case $ M=49,N=100 $ performs better than $ M=100,N=49 $.}  In Fig. \ref{fig5}.(b), we show the achievable rate versus the number of transmit antennas $ N_{t} $ for varying number of meta-atoms $ M,N $. The rate saturates for large values of data streams due to increasing multiplexing gain. Moreover, as the number of meta-atoms increase, the rate is improved.


\section{Conclusion} \label{Conclusion} 
In this paper, we proposed a SIM-assisted HMIMO communication system, where the transmitter and the receiver are implemented by SIMs that enhance the wave-based analog precoding and combining. Specifically, we formulated the achievable rate problem, and we proposed an efficient gradient ascent algorithm that optimizes the covariance matrix of the transmitted signal and the SIM phase shifts at both sides of the transceiver simultaneously. Moreover, we obtained a Lipschitz constant guaranteeing the convergence of the proposed iterative algorithm. Numerical results confirmed that the proposed algorithm requires a lower number of iterations compared to the usually used AO approach. Finally, we showed the outperformance of the architecture with respect to its single-RIS counterpart and conventional MIMO system.

	\begin{appendices}
	\section{Proof of Lemma~\ref{lemmaGradient}}\label{lem1}
For the derivation of $ \nabla_{\bQ}f(\bQ,\bphi_{l},\bpsi_{k}) $, we start by deriving first its differential. We have
	\begin{align}
	d(f(\bQ,\bphi_{l},\bpsi_{k}))&=\tr(\bK(\bQ,\bphi_{l},\bpsi_{k})d(\bar{\bH}\bQ\bar{\bH}^{\H}))\label{differentialQ1}\\
	&=\tr(\bar{\bH}^{\H}\bK(\bQ,\bphi_{l},\bpsi_{k})\bar{\bH} d(\bQ)),\label{differentialQ2}
	\end{align}
	where, in \eqref{differentialQ1}, we have applied the property $ d(\det(\bX))=\det(\bX)\tr(\bX^{-1}d(\bX)) $, and in \eqref{differentialQ2}, we have applied the property $ \tr(\bA \bB)=\tr( \bB\bA) $. Note that we have denoted $ \bK(\bQ,\bphi_{l},\bpsi_{k})=\left(\Id+\bar{\bH}\bQ\bar{\bH}^{\H}\right)^{-1} $. 
	
	From \eqref{differentialQ2}, we obtain $ \nabla_{\bQ}f(\bQ,\bphi_{l},\bpsi_{k}) $ by applying the suitable property from \cite[Table 3.2]{hjorungnes:2011}. Hence, we obtain
	\begin{align}
	\nabla_{\bQ}f(\bQ,\bphi_{l},\bpsi_{k})&=\frac{\partial}{\bQ^{*}}f(\bQ,\bphi_{l},\bpsi_{k})\\
	&=\left(\frac{\partial}{\bQ}f(\bQ,\bphi_{l},\bpsi_{k})\right)^{\T}\\
	&	=\left(\bar{\bH}^{\T}\left(\Id+\bar{\bH}^{*}\bQ^{\T}\bar{\bH}^{\T}\right)^{-1}\bar{\bH}^{*}\right)^{\T}\\
		&=\bar{\bH}^{\H}\left(\Id+\bar{\bH}\bQ\bar{\bH}^{\H}\right)^{-1}\bar{\bH},
			\end{align}
		where, in the last step, we have used that $ \bQ^{*}=\bQ^{\T} $.

For the derivation of $ \nabla_{\bphi_{l}}f(\bQ,\bphi_{l},\bpsi_{k}) $, we focus on its differential. We have
	\begin{align}
	&	d(f(\bQ,\bphi_{l},\bpsi_{k}))
		\!=\!\tr(\bK(\bQ,\bphi_{l},\bpsi_{k})(d(\bar{\bH}) \bQ\bar{\bH}^{\H}\!+\!\bar{\bH} \bQ d(\bar{\bH}^{\H}))),\label{differentialPhi1}\\
			&=\tr(\bQ\bar{\bH}^{\H}\bK(\bQ,\bphi_{l},\bpsi_{k})d(\bar{\bH}) +\bK(\bQ,\bphi_{l},\bpsi_{k})\bar{\bH} \bQ d(\bar{\bH}^{\H})).\label{differentialPhi2}
	\end{align}
	
From \eqref{EquivalentChannel}, it is easy to check that 
	\begin{align}
		d(\bar{\bH})=\bZ\bar{\bG} d(\bP),\label{differentialPhi3}
	\end{align}
where $ \bar{\bG}= \bG/\sqrt{N_{0}}$. Also, the differential of \eqref{TransmitterSIM} can be written as
\begin{align}
		d(\bP)&=\bPhi^{L}\bW^{L}\cdots\bPhi^{l+1}\bW^{l+1}d(\bPhi^{l})\bW^{l}\bPhi^{l-1}\nn\\
		&\times\bW^{l-1}\cdots\bPhi^{1}\bW^{1}. \label{differentialPhi4}
\end{align}
Substituting \eqref{differentialPhi4} and \eqref{differentialPhi3} into \eqref{differentialPhi2}, we obtain
\begin{align}
		d(f(\bQ,\bphi_{l},\bpsi_{k}))		&=\tr(\bA_{l} d(\bPhi^{l}) +\bA_{l}^{\H} d((\bPhi^{l})^{\H})),\label{differentialPhi5}
\end{align}
where 
\begin{align}
	\bA_{l}&=\bW^{l}\bPhi^{l-1}\bW^{l-1}\cdots \bPhi^{1}\bW^{1}\bQ\bar{\bH}^{\H}\bK \bZ\bar{\bG}\bPhi^{L}\nn\\
	&\times\bW^{L}\cdots\bPhi^{l+1}\bW^{l+1}.
\end{align}

Now, using the property $ \tr\left(\bA\bB\right)=(\left(\diag\left(\mathbf{A}\right)\right)^{\T}d(\diag(\bB))$ for any matrices $ \bA,\bB $ with $ \bB $ being a diagonal matrix, \eqref{differentialPhi5} becomes
\begin{align}
d(f(\bQ,\bphi_{l},\bpsi_{k}))		&=(\diag(\bA_{l}))^{\T} d(\bphi^{l})\nn\\
& +(\diag( \bA_{l}^{\H}))^{\T} d(\bphi^{l*})\label{differentialPhi6}.
\end{align}

	From \eqref{differentialPhi6}, we can conclude that 
\begin{align}
		\nabla_{\bphi_{l}}f(\bQ,\bphi_{l},\bpsi_{k})&=\frac{\partial}{\partial{(\bphi^{l*})}}	
		\Big((\diag(\bA_{l}))^{\T} d(\bphi^{l})\nn\\
		& +(\diag( \bA_{l}^{\H}))^{\T} d(\bphi^{l*})\Big)\nn\\
	&=\diag( \bA_{l}^{\H})\label{differentialPhi7}.
\end{align}

Similar to the derivation of \eqref{differentialPhi7}, we obtain
\begin{align}
	\nabla_{\bpsi_{k}}f(\bQ,\bphi_{l},\bpsi_{k})
	&=\diag( \bC_{k}^{\H}),\label{differentialPhi8}
\end{align}
where 
\begin{align}
	\bC_{k}&=\bU^{k}\bPsi^{k-1}\bU^{k-1}\cdots \bPsi^{1}\bU^{1}\bK\bar{\bH}\bQ\bP^{\H} \bar{\bG}^{\H}\bPsi^{K} \nn\\
	&\times\bU^{K}\cdots \bPsi^{k+1}\bU^{k+1}.
\end{align}

	\section{Proof  of Proposition~\ref{proposition1}}\label{Prop1}
	During the proof, we will use the following inequalities
	\begin{align}
	\|\bA\bB\|&\le \lambda_{\mathrm{max}}(\bA)\|\bB\| \label{norm1},\\
	\|\bA\bB\bC\|&\le \lambda_{\mathrm{max}}(\bA)\lambda_{\mathrm{max}}(\bC)\|\bB\| \label{norm2},
\end{align}
where $ \lambda_{\mathrm{max}}(\bX) $ is the largest singular value of $ \bX $.	Also, we have $ \lambda_{\mathrm{max}}(\bQ)\le P $, and $\bK(\bQ,\bphi_{l},\bpsi_{k})=\left(\Id+\bar{\bH}\bQ\bar{\bH}^{\H}\right)^{-1} \preceq \Id $, which gives
\begin{align}
	\lambda_{\mathrm{max}}(\bK(\bQ,\bphi_{l},\bpsi_{k}))
\le 1 \label{kappa}.
\end{align}

	Making use of \eqref{gradient1}, we can write
	\begin{align}
		&\|  \nabla_{\bQ}f(\bQ^{1},\bphi_{l}^{1},\bpsi_{k}^{1})-  \nabla_{\bQ}f(\bQ^{2},\bphi_{l}^{2},\bpsi_{k}^{2})\|\nn\\
		&
		=\|\bar{\bH}_{1}^{\H}	\bK(\bQ^{1},\bphi_{l}^{1},\bpsi_{k}^{1})\bar{\bH}_{1}-\bar{\bH}_{2}^{\H}	\bK(\bQ^{2},\bphi_{l}^{2},\bpsi_{k}^{2})\bar{\bH}_{2}\|\nn\\
		&\le \|(\bP^{1})^{\H}\bG^{\H}(\bZ^{1})^{\H}	\bK(\bQ^{1},\bphi_{l}^{1},\bpsi_{k}^{1})\bZ^{1}\bG\bP^{1}-\nn\\&	(\bP^{1})^{\H}\bG^{\H}(\bZ^{1})^{\H}	\bK(\bQ^{1},\bphi_{l}^{1},\bpsi_{k}^{1})\bZ^{1}\bG\bP^{2}\|\nn\\
		&+\|(\bP^{1})^{\H}\bG^{\H}(\bZ^{1})^{\H}	\bK(\bQ^{1},\bphi_{l}^{1},\bpsi_{k}^{1})\bZ^{1}\bG\bP^{2}\nn\\&	-(\bP^{2})^{\H}\bG^{\H}(\bZ^{2})^{\H}	\bK(\bQ^{2},\bphi_{l}^{2},\bpsi_{k}^{2})\bZ^{2}\bG\bP^{2}\|\label{gradientQ},
	\end{align}
	where, in \eqref{gradientQ}, we have used \eqref{EquivalentChannel}. 
	
	The first term of \eqref{gradientQ} can be upper-bounded as
	\begin{align}
		&	\|(\bP^{1})^{\H}\bG^{\H}(\bZ^{1})^{\H}	\bK(\bQ^{1},\bphi_{l}^{1},\bpsi_{k}^{1})\bZ^{1}\bG\bP^{1}\nn\\
		&-	(\bP^{1})^{\H}\bG^{\H}(\bZ^{1})^{\H}	\bK(\bQ^{1},\bphi_{l}^{1},\bpsi_{k}^{1})\bZ^{1}\bG\bP^{2}\|
		\nn\\
		&\le \frac{d^{2}b^{2} f}{c^{2}} \|\bP^{1}-\bP^{2}\|\label{gradientQ1},
	\end{align}
	where, in \eqref{gradientQ1}, we have used \eqref{bg}, \eqref{Eqf}, and \eqref{kappa}.
	
	We upper bound the second term in \eqref{gradientQ1} as
	\begin{align}
		&\|(\bP^{1})^{\H}\bG^{\H}(\bZ^{1})^{\H}	\bK(\bQ^{1},\bphi_{l}^{1},\bpsi_{k}^{1})\bZ^{1}\bG\bP^{2}\nn\\&	-(\bP^{2})^{\H}\bG^{\H}(\bZ^{2})^{\H}	\bK(\bQ^{2},\bphi_{l}^{2},\bpsi_{k}^{2})\bZ^{2}\bG\bP^{2}\| \nn\\
		&\le bf \|(\bP^{1})^{\H}\bG^{\H}(\bZ^{1})^{\H}	\bK(\bQ^{1},\bphi_{l}^{1},\bpsi_{k}^{1})\bZ^{1}	\nn\\&-(\bP^{2})^{\H}\bG^{\H}(\bZ^{2})^{\H}	\bK(\bQ^{2},\bphi_{l}^{2},\bpsi_{k}^{2})\bZ^{2}\| \label{gradientQ2}\\
		&\le bf \big( \|(\bP^{1})^{\H}\bG^{\H}(\bZ^{1})^{\H}	\bK(\bQ^{1},\bphi_{l}^{1},\bpsi_{k}^{1})\bZ^{1}
		\nn\\&-(\bP^{2})^{\H}\bG^{\H}(\bZ^{1})^{\H}	\bK(\bQ^{1},\bphi_{l}^{1},\bpsi_{k}^{1})\bZ^{1} \|\nn\\
		&+ \|(\bP^{2})^{\H}\bG^{\H}(\bZ^{1})^{\H}	\bK(\bQ^{1},\bphi_{l}^{1},\bpsi_{k}^{1})\bZ^{1}
		\nn\\&-(\bP^{2})^{\H}\bG^{\H}(\bZ^{2})^{\H}	\bK(\bQ^{2},\bphi_{l}^{2},\bpsi_{k}^{2})\bZ^{2}\|\big) \label{gradientQ3},
	\end{align}
	where, in \eqref{gradientQ2}, we have used \eqref{bg} and \eqref{Eqd}. 
	
	The first term of \eqref{gradientQ3} results in
	\begin{align}
	&\!\!\!	\|(\bP^{1})^{\H}\bG^{\H}(\bZ^{1})^{\H}	\bK(\bQ^{1},\bphi_{l}^{1},\bpsi_{k}^{1})\bZ^{1}
		\nn\\&\!\!-(\bP^{2})^{\H}\bG^{\H}(\bZ^{1})^{\H}	\bK(\bQ^{1},\bphi_{l}^{1},\bpsi_{k}^{1})\bZ^{1} \|\le \frac{b d^{2}}{c^{2}} \|\bP^{1}- \bP^{2}\|,\label{gradientQ30}
	\end{align}
	where we have used \eqref{kappa}.
	
	The second term in \eqref{gradientQ3} becomes 
	\begin{align}
		&\|(\bP^{2})^{\H}\bG^{\H}(\bZ^{1})^{\H}	\bK(\bQ^{1},\bphi_{l}^{1},\bpsi_{k}^{1})\bZ^{1}\nn\\
		&
		-(\bP^{2})^{\H}\bG^{\H}(\bZ^{2})^{\H}	\bK(\bQ^{2},\bphi_{l}^{2},\bpsi_{k}^{2})\bZ^{2}\|\nn\\
		&\le bf \|(\bZ^{1})^{\H}	\bK(\bQ^{1},\bphi_{l}^{1},\bpsi_{k}^{1})\bZ^{1}-(\bZ^{2})^{\H}	\bK(\bQ^{2},\bphi_{l}^{2},\bpsi_{k}^{2})\bZ^{2}\| \label{gradientQ31}\\
		&\le bf \big(\|(\bZ^{1})^{\H}	\bK(\bQ^{1},\bphi_{l}^{1},\bpsi_{k}^{1})\bZ^{1}-(\bZ^{2})^{\H}	\bK(\bQ^{1},\bphi_{l}^{1},\bpsi_{k}^{1})\bZ^{1}\|
		\nn\\
		&+\|(\bZ^{2})^{\H}	\bK(\bQ^{1},\bphi_{l}^{1},\bpsi_{k}^{1})\bZ^{1}
		-(\bZ^{2})^{\H}	\bK(\bQ^{2},\bphi_{l}^{2},\bpsi_{k}^{2})\bZ^{2}\|\big) \label{gradientQ32}\\
		&\le \frac{2 bf d}{c} \big(\|\bZ^{1}-\bZ^{2}\|, \label{gradientQ33}
	\end{align}
	where, in \eqref{gradientQ31}, we have used \eqref{bg} and \eqref{kappa}.

Substitution of \eqref{gradientQ30} and \eqref{gradientQ33} into \eqref{gradientQ2} gives
\begin{align}
&\|(\bP^{1})^{\H}\bG^{\H}(\bZ^{1})^{\H}	\bK(\bQ^{1},\bphi_{l}^{1},\bpsi_{k}^{1})\bZ^{1}\bG\bP^{2}	\nn\\&-(\bP^{2})^{\H}\bG^{\H}(\bZ^{2})^{\H}	\bK(\bQ^{2},\bphi_{l}^{2},\bpsi_{k}^{2})\bZ^{2}\bG\bP^{2}\| \nn\\
&\le 
bf \big( \frac{b d^{2}}{c^{2}} \|\bP^{1}- \bP^{2}\|+ \frac{2 bf d}{c} \big(\|\bZ^{1}-\bZ^{2}\|\big). \label{gradientQ20}
\end{align}

Now, inserting \eqref{gradientQ1} and \eqref{gradientQ20} into \eqref{gradientQ}, we obtain
\begin{align}
&\|  \nabla_{\bQ}f(\bQ^{1},\bphi_{l}^{1},\bpsi_{k}^{1})-  \nabla_{\bQ}f(\bQ^{2},\bphi_{l}^{2},\bpsi_{k}^{2})\|\nn\\
&	\le bf \big( \frac{2 b d^{2}}{c^{2}} \|\bP^{1}- \bP^{2}\|+ \frac{2 bf d}{c} \big(\|\bZ^{1}-\bZ^{2}\|\big).
\end{align}
	
Using \eqref{gradient2}, we obtain
\begin{align}
&\|  \nabla_{\bphi_{l}}f(\bQ^{1},\bphi_{l}^{1},\bpsi_{k}^{1})-  \nabla_{\bphi_{l}}f(\bQ^{2},\bphi_{l}^{2},\bpsi_{k}^{2})\nn\\&\|\le\|\bK(\bQ^{1},\bphi_{l}^{1},\bpsi_{k}^{1}) \bA_{l,1}^{\H}-\bK(\bQ^{2},\bphi_{l}^{2},\bpsi_{k}^{2}) \bA_{l,2}^{\H}\|\label{proof1}\\
&\le a_{l} b \| \bQ^{1}\bar{\bH}_{1}^{\H}\bZ^{1}- \bQ^{2}\bar{\bH}_{2}^{\H}\bZ^{2}\|\label{proof2}\\
&= a_{l} b \| \bQ^{1}\bar{\bH}_{1}^{\H}\bZ^{1}- \bQ^{2}\bar{\bH}_{1}^{\H}\bZ^{1}+\bQ^{2}\bar{\bH}_{1}^{\H}\bZ^{1}- \bQ^{2}\bar{\bH}_{2}^{\H}\bZ^{2}\|\label{proof3}\\
&\le a_{l} b \big(\| \bQ^{1}\bar{\bH}_{1}^{\H}\bZ^{1}- \bQ^{2}\bar{\bH}_{1}^{\H}\bZ^{1}\|+\|\bQ^{2}\bar{\bH}_{1}^{\H}\bZ^{1}- \bQ^{2}\bar{\bH}_{2}^{\H}\bZ^{2}\|\big)\label{proof4},
\end{align}	
where, in \eqref{proof1}, we have used that $ \|\diag(\bX)\|\le \|\bX\| $. In \eqref{proof2}, we have used \eqref{kappa}, \eqref{al}, and \eqref{bg}.
	
We upper-bound the first term on the right-hand side of \eqref{proof4} as
\begin{align}
	\| \bQ^{1}\bar{\bH}_{1}^{\H}\bZ^{1}- \bQ^{2}\bar{\bH}_{1}^{\H}\bZ^{1}\|\le c\|\bQ^{1}-\bQ^{2}\|,\label{FirstTerm}
\end{align}
	where we have used \eqref{Eqd}. 
	
The second term of \eqref{proof4} can be written
\begin{align}
	&\|\bQ^{2}\bar{\bH}_{1}^{\H}\bZ^{1}- \bQ^{2}\bar{\bH}_{2}^{\H}\bZ^{2}\|\le P\|\bar{\bH}_{1}^{\H}\bZ^{1}- \bar{\bH}_{2}^{\H}\bZ^{2}\|\label{secondTerm1}\\
	&= P\|\bar{\bH}_{1}^{\H}\bZ^{1}-\bar{\bH}_{1}^{\H}\bZ^{2}+\bar{\bH}_{1}^{\H}\bZ^{2}- \bar{\bH}_{2}^{\H}\bZ^{2}\|\label{secondTerm2}\\
	&\le P\big(\|\bar{\bH}_{1}^{\H}\bZ^{1}-\bar{\bH}_{1}^{\H}\bZ^{2}\|+\|\bar{\bH}_{1}^{\H}\bZ^{2}- \bar{\bH}_{2}^{\H}\bZ^{2}\|\big)\label{secondTerm3}.
\end{align}
The first term in \eqref{secondTerm3} can be written as
\begin{align}
	\|\bar{\bH}_{1}^{\H}\bZ^{1}-\bar{\bH}_{1}^{\H}\bZ^{2}\|\le c\|\bZ^{1}-\bZ^{2}\|,\label{secondTerm30}
\end{align}
where we have used \eqref{Eqc}. 

Next, the second term in \eqref{secondTerm3} becomes
\begin{align}
&\|\bar{\bH}_{1}^{\H}\bZ^{2}- \bar{\bH}_{2}^{\H}\bZ^{2}\|\le \frac{d}{c}\|\bar{\bH}_{1}- \bar{\bH}_{2}\|\label{secondTerm31}\\
&=\frac{d}{c N_{0}}\|\bZ^{1}\bG\bP^{1}- \bZ^{2}\bG\bP^{2}\|\label{secondTerm32}\\
&=\frac{d}{c N_{0}}\|\bZ^{1}\bG\bP^{1}-\bZ^{2}\bG\bP^{1}+\bZ^{2}\bG\bP^{1}- \bZ^{2}\bG\bP^{2}\|\label{secondTerm33}\\
&=\frac{d}{c N_{0} }\big(\|\bZ^{1}\bG\bP^{1}-\bZ^{2}\bG\bP^{1}\|+\|\bZ^{2}\bG\bP^{1}- \bZ^{2}\bG\bP^{2}\|\big)\label{secondTerm34},
\end{align}
where, in \eqref{secondTerm31}, we have used \eqref{Eqd}. In \eqref{secondTerm32}, we have inserted \eqref{EquivalentChannel}.

The first term in \eqref{secondTerm34} becomes
\begin{align}
\|\bZ^{1}\bG\bP^{1}-\bZ^{2}\bG\bP^{1}\|\le b f	\|\bZ^{1}-\bZ^{2}\|\label{secondTerm341}
\end{align}
where we have used \eqref{Eqf}. The second term in \eqref{secondTerm34} yields
\begin{align}
	\|\bZ^{2}\bG\bP^{1}- \bZ^{2}\bG\bP^{2}\|\le b \frac{d}{c} \|\bP^{1}- \bP^{2}\|\label{secondTerm342}.
\end{align}

By substituting \eqref{secondTerm341} and \eqref{secondTerm342} into \eqref{secondTerm34}, we obtain
\begin{align}
	\|\bar{\bH}_{1}^{\H}\bZ^{2}- \bar{\bH}_{2}^{\H}\bZ^{2}\|&\le\frac{b d}{c N_{0}}\big( f	\|\bZ^{1}-\bZ^{2}\|+ \frac{d}{c } \|\bP^{1}- \bP^{2}\|\big)\label{secondTerm343}.
\end{align}

Inserting \eqref{secondTerm30} and \eqref{secondTerm343} into \eqref{secondTerm3}
\begin{align}
		\|\bQ^{2}\bar{\bH}_{1}^{\H}\bZ^{1}- \bQ^{2}\bar{\bH}_{2}^{\H}\bZ^{2}\|
&\le P\left(c+ \frac{bdf}{c N_{0}}\right)\|\bZ^{1}-\bZ^{2}\|
\nn\\
&+P\frac{bd^{2}}{c^{2} N_{0}} \|\bP^{1}- \bP^{2}\|\label{secondTerm40}.
\end{align}

Finally, substituting \eqref{secondTerm40} and \eqref{FirstTerm} into \eqref{proof4}, we derive 
\begin{align}
&	\|  \nabla_{\bphi_{l}}f(\bQ^{1},\bphi_{l}^{1},\bpsi_{k}^{1})-  \nabla_{\bphi_{l}}f(\bQ^{2},\bphi_{l}^{2},\bpsi_{k}^{2})\|\le a_{l} b c |\bQ^{1}-\bQ^{2}\|
	\nn\\&+a_{l} b P\left(c+ \frac{bdf}{c N_{0}}\right)\|\bZ^{1}-\bZ^{2}\|\nn
	\\&+a_{l} P\frac{b^{2}d^{2}}{c^{2}N_{0}} \|\bP^{1}- \bP^{2}\|.
\end{align}

The inequality regarding $ 	\nabla_{\bpsi_{k}}f(\bQ,\bphi_{l},\bpsi_{k}) $ can be written as
\begin{align}
&\|		\nabla_{\bpsi_{k}}f(\bQ^{1},\bphi_{l}^{1},\bpsi_{k}^{1}) -\nabla_{\bpsi_{k}}f(\bQ^{2},\bphi_{l}^{2},\bpsi_{k}^{2})\|\nn\\
&\le\|\bK(\bQ^{1},\bphi_{l}^{1},\bpsi_{k}^{1}) \bC_{k,1}^{\H}-\bK(\bQ^{2},\bphi_{l}^{2},\bpsi_{k}^{2} \bC_{k,2}^{\H})\|.
\end{align}

In a similar way, we can upper bound the previous inequality. We omit the details due to limited space. Hence, we conclude the proof.

	\section{Proof  of Theorem~\ref{Theorem1}}\label{Theo1}
According to Lemma \eqref{lemmaGradient}, we can write
\begin{align}
&\|  \nabla_{\bQ}f(\bQ^{1},\bphi_{l}^{1},\bpsi_{k}^{1})-  \nabla_{\bQ}f(\bQ^{2},\bphi_{l}^{2},\bpsi_{k}^{2})\|^{2}
\nn\\&+
\|  \nabla_{\bphi_{l}}f(\bQ^{1},\bphi_{l}^{1},\bpsi_{k}^{1})-  \nabla_{\bphi_{l}}f(\bQ^{2},\bphi_{l}^{2},\bpsi_{k}^{2})\|^{2}\nn\\
&+	\|		\nabla_{\bpsi_{k}}f(\bQ^{1},\bphi_{l}^{1},\bpsi_{k}^{1}) -\nabla_{\bpsi_{k}}f(\bQ^{2},\bphi_{l}^{2},\bpsi_{k}^{2})\|
\nn\\
&\le 	\Lambda_{\bQ}^{2} \|\bQ^{1}-\bQ^{2}\|^{2}+ \Lambda_{\bphi_{l}}^{2}\|\bP^{1}- \bP^{2}\|^{2}+ \Lambda_{\bpsi_{k}}^{2} \|\bZ^{1}-\bZ^{2}\|^{2}\nn\\
&\le \max(\Lambda_{\bQ}^{2},\Lambda_{\bphi_{l}}^{2},\Lambda_{\bpsi_{k}}^{2}) \big(	 \|\bQ^{1}-\bQ^{2}\|^{2}+ \|\bP^{1}- \bP^{2}\|^{2}
\nn\\
&+ \|\bZ^{1}-\bZ^{2}\|^{2}\big).\label{Theorem1eq}
\end{align}
Applying the square root on both sides of \eqref{Theorem1eq} together with the inequality
\begin{align}
	&2 \|\bQ^{1}-\bQ^{2}\|\cdot \|\bP^{1}- \bP^{2}\|+	2 \|\bZ^{1}-\bZ^{2}\|\cdot \|\bP^{1}- \bP^{2}\|
	\nn\\&+2	 \|\bQ^{1}-\bQ^{2}\|\cdot \|\bZ^{1}-\bZ^{2}\|\nn\\
	&\le \|\bQ^{1}-\bQ^{2}\|^{2}+\|\bP^{1}- \bP^{2}\|+\|\bZ^{1}-\bZ^{2}\|,
\end{align}
we observe that $\max(\Lambda_{\bQ}^{2},\Lambda_{\bphi_{l}}^{2},\Lambda_{\bpsi_{k}}^{2}) $ is a Lipschitz constant of the gradient of $ f(\bQ,\bphi_{l},\bpsi_{k}) $, which concludes the proof.

\section{Proof  of Theorem~\ref{Theorem2}}\label{Theo2}

The point $ \bQ^{n+1} $ can be projected onto $\mathcal{ Q} $ according to the following equation.
\begin{align}
\bQ^{n+1}&=\arg \min_{\bQ \in \mathcal{Q}}	\|\bQ-\bQ^{n}-\mu_{n}^{1} \nabla_{\mathcal{Q}}f(\bQ^{n},\bphi_{l}^{n},\bpsi_{k}^{n}) \|^{2}\nn\\
&=\arg \max_{\bQ \in \mathcal{Q}}\langle\nabla_{\mathcal{Q}}f(\bQ^{n},\bphi_{l}^{n},\bpsi_{k}^{n}), \bQ-\bQ^{n}\rangle \nn\\&-\frac{1}{ 2 \mu_{n}^{1} }\|\bQ-\bQ^{n}\|^{2}.\label{optimality}
\end{align}
 In the case of $ \bQ-\bQ^{n} $, we meet that the objective is $ 0 $, which means that
 \begin{align}
\!\!\!\!\!\langle\nabla_{\mathcal{Q}}f(\bQ^{n},\bphi_{l}^{n},\bpsi_{k}^{n}), \bQ-\bQ^{n}\rangle -\frac{1}{ 2 \mu_{n}^{1} }\|\bQ-\bQ^{n}\|^{2}\ge 0. 	
 \end{align}

Similarly, we obtain
\textcolor{black}{\begin{align}
&	\langle\nabla_{\bphi_{l}}f(\bQ^{n},\bphi_{l}^{n},\bpsi_{k}^{n}), \bphi_{l}-\bphi_{l}^{n}\rangle -\frac{1}{ 2 \mu_{n}^{2} }\|\bphi_{l}-\bphi_{l}^{n}\|^{2}\ge 0,\\
&	\langle\nabla_{\bpsi_{k}}f(\bQ^{n},\bphi_{l}^{n},\bpsi_{k}^{n}), \bpsi_{k}-\bpsi_{k}^{n}\rangle -\frac{1}{ 2 \mu_{n}^{3} }\|\bpsi_{k}-\bpsi_{k}^{n}\|^{2}\ge 0. 	
\end{align}}

Now, applying the following inequality
\begin{align}
	f(\by)\ge f(\by)+\langle	\nabla f(\bx),\by-\bx\rangle-\frac{\Lambda}{2}\|\by-\bx\|^{2},\label{smooth}
\end{align}
which holds for any $ \Lambda $-smooth function $ f(\bx) $, we obtain
\textcolor{black}{\begin{align}
	&f(\bQ^{n+1},\bphi_{l}^{n+1},\bpsi_{k}^{n+1})\ge f(\bQ^{n},\bphi_{l}^{n},\bpsi_{k}^{n})
	\nn\\
	&+\langle\nabla_{\mathcal{Q}}f(\bQ^{n},\bphi_{l}^{n},\bpsi_{k}^{n}), \bQ-\bQ^{n}\rangle \nn\\
	&+	\langle\nabla_{\bphi_{l}}f(\bQ^{n},\bphi_{l}^{n},\bpsi_{k}^{n}), \bphi_{l}-\bphi_{l}^{n}\rangle
	\nn\\&+	\langle\nabla_{\bpsi_{k}}f(\bQ^{n},\bphi_{l}^{n},\bpsi_{k}^{n}), \bpsi_{k}-\bpsi_{k}^{n}\rangle 
	\nn\\
	&-\frac{\Lambda}{2}\|\bQ^{n+1}-\bQ^{n}\|^{2}-\frac{\Lambda}{2}\|\bphi_{l}^{n+1}-\bphi_{l}^{n}\|^{2}-\frac{\Lambda}{2}\|\bpsi_{k}^{n+1}-\bpsi_{k}^{n}\|^{2}\nn\\
	&\ge f(\bQ^{n},\bphi_{l}^{n},\bpsi_{k}^{n})+\left(\frac{1}{2 \mu_{n}^{1}}-\frac{\Lambda}{2}\right)\|\bQ^{n+1}-\bQ^{n}\|^{2}
	\nn\\&+\left(\frac{1}{2 \mu_{n}^{2}}-\frac{\Lambda}{2}\right)\|\bphi_{l}^{n+1}-\bphi_{l}^{n}\|^{2}\nn\\
	&+\left(\frac{1}{2 \mu_{n}^{3}}-\frac{\Lambda}{2}\right)\|\bpsi_{k}^{n+1}-\bpsi_{k}^{n}\|^{2},\label{smoothL}
\end{align}}
where, in the first inequality, we have applied Theorem 1.

From the above equation, when $ \mu_{n}^{q} < \frac{1}{\Lambda} $ for $ q=1,2,3 $, we observe that $ f(\bQ^{n+1},\bphi_{l}^{n+1},\bpsi_{k}^{n+1})\ge f(\bQ^{n},\bphi_{l}^{n},\bpsi_{k}^{n}) $. Now, we denote $ f^{\star} $ the value of $ f $ at all accumulation points, and we write \eqref{smoothL} as
\textcolor{black}{\begin{align}
	&f(\bQ^{n+1},\bphi_{l}^{n+1},\bpsi_{k}^{n+1})- f(\bQ^{n},\bphi_{l}^{n},\bpsi_{k}^{n})\ge \left(\frac{1}{2 \mu_{n}^{1}}-\frac{\Lambda}{2}\right)\nn\\
	&\times \|\bQ^{n+1}-\bQ^{n}\|^{2}+\left(\frac{1}{2 \mu_{n}^{2}}-\frac{\Lambda}{2}\right)\|\bphi_{l}^{n+1}-\bphi_{l}^{n}\|^{2}\nn\\
	&+\left(\frac{1}{2 \mu_{n}^{3}}-\frac{\Lambda}{2}\right)\|\bpsi_{k}^{n+1}-\bpsi_{k}^{n}\|^{2},
\end{align}}
which gives 
\textcolor{black}{\begin{align}
&	\infty > f^{\star}-f(\bQ^{1},\bphi_{l}^{1},\bpsi_{k}^{1})>\sum _{n=1}^{\infty}\Bigg(\left(\frac{1}{2 \mu_{n}^{1}}-\frac{\Lambda}{2}\right)\nn\\
&	\times \|\bQ^{n+1}-\bQ^{n}\|^{2}+\left(\frac{1}{2 \mu_{n}^{2}}-\frac{\Lambda}{2}\right)\|\bphi_{l}^{n+1}-\bphi_{l}^{n}\|^{2}\nn\\
&+\left(\frac{1}{2 \mu_{n}^{3}}-\frac{\Lambda}{2}\right)\|\bpsi_{k}^{n+1}-\bpsi_{k}^{n}\|^{2}\Bigg).
\end{align}}
\textcolor{black}{Given that $ \mu_{n}^{q} < \frac{1}{\Lambda} $ for $ q=1,2,3 $, we result in}
\begin{align}
	\|\bQ^{n+1}-\bQ^{n}\|&\to 0,\\
	\|\bphi_{l}^{n+1}-\bphi_{l}^{n}\|&\to 0,\\
	\|\bpsi_{k}^{n+1}-\bpsi_{k}^{n}\|&\to 0.
\end{align}

The condition in \eqref{optimality}, concerning optimality, can be written as
\textcolor{black}{\begin{align}
	&\langle \frac{1}{\mu_{n}^{i}} \left(\bQ^{n+1}-\bQ^{n} \right)- \nabla_{\mathcal{Q}}f(\bQ^{n},\bphi_{l}^{n},\bpsi_{k}^{n}),\bQ-\bQ^{n+1} \rangle \nn\\
	&\le 0, \forall \bQ \in \mathcal{Q}.\label{optimality1}
\end{align}}

Similar inequalities hold for the other two parameters $ \bphi_{l} $ and $ \bpsi_{k} $.

By setting $ n \to \infty $ in \eqref{optimality1}, we obtain
\begin{align}
	\langle - \nabla_{\mathcal{Q}}f(\bQ^{*},\bphi_{l}^{*},\bpsi_{k}^{*}),\bQ-\bQ^{\star} \rangle \le 0, \forall \bQ \in \mathcal{Q}
\end{align}
due to the continuity of the gradient of $ f(\bQ^{n},\bphi_{l}^{n},\bpsi_{k}^{n}) $, i.e., $ \nabla_{\mathcal{Q}}f(\bQ^{n},\bphi_{l}^{n},\bpsi_{k}^{n})\to \nabla_{\mathcal{Q}}f(\bQ^{*},\bphi_{l}^{*},\bpsi_{k}^{*}) $.

In the case of $ \bphi_{l} $ and $ \bpsi_{k} $, similar observations hold, which signifies that $ (\bQ^{*},\bphi_{l}^{*},\bpsi_{k}^{*}) $ is a critical point of Problem $ 	(\mathcal{P}) $, and this concludes the proof.
	\end{appendices}
	\bibliographystyle{IEEEtran}

	\bibliography{IEEEabrv,bibl}

\begin{thebibliography}{10}
\providecommand{\url}[1]{#1}
\csname url@samestyle\endcsname
\providecommand{\newblock}{\relax}
\providecommand{\bibinfo}[2]{#2}
\providecommand{\BIBentrySTDinterwordspacing}{\spaceskip=0pt\relax}
\providecommand{\BIBentryALTinterwordstretchfactor}{4}
\providecommand{\BIBentryALTinterwordspacing}{\spaceskip=\fontdimen2\font plus
\BIBentryALTinterwordstretchfactor\fontdimen3\font minus
  \fontdimen4\font\relax}
\providecommand{\BIBforeignlanguage}[2]{{%
\expandafter\ifx\csname l@#1\endcsname\relax
\typeout{** WARNING: IEEEtran.bst: No hyphenation pattern has been}%
\typeout{** loaded for the language `#1'. Using the pattern for}%
\typeout{** the default language instead.}%
\else
\language=\csname l@#1\endcsname
\fi
#2}}
\providecommand{\BIBdecl}{\relax}
\BIBdecl

\bibitem{Letaief2019}
K.~B. Letaief \emph{et~al.}, ``The roadmap to {6G}: {AI} empowered wireless
  networks,'' \emph{IEEE Commun. Mag.}, vol.~57, no.~8, pp. 84--90, 2019.

\bibitem{Dang2020}
S.~Dang \emph{et~al.}, ``What should {6G} be?'' \emph{Nature Electronics},
  vol.~3, no.~1, pp. 20--29, 2020.

\bibitem{Andrews2014}
J.~G. Andrews \emph{et~al.}, ``What will 5{G} be?'' \emph{IEEE J. Sel. Areas
  Commun.}, vol.~32, no.~6, pp. 1065--1082, 2014.

\bibitem{DiRenzo2020}
M.~Di~Renzo \emph{et~al.}, ``Smart radio environments empowered by
  reconfigurable intelligent surfaces: {How} it works, state of research, and
  the road ahead,'' \emph{IEEE J. Sel. Areas Commun.}, vol.~38, no.~11, pp.
  2450--2525, 2020.

\bibitem{Wu2020}
Q.~{Wu} and R.~{Zhang}, ``Towards smart and reconfigurable environment:
  Intelligent reflecting surface aided wireless network,'' \emph{IEEE Commun.
  Mag.}, vol.~58, no.~1, pp. 106--112.

\bibitem{Papazafeiropoulos2021}
A.~Papazafeiropoulos \emph{et~al.}, ``Intelligent reflecting surface-assisted
  {MU-MISO} systems with imperfect hardware: {Channel} estimation and
  beamforming design,'' \emph{IEEE Trans. Wireless Commun.}, vol.~21, no.~3,
  pp. 2077--2092, 2021.

\bibitem{Papazafeiropoulos2023c}
------, ``Cooperative {RIS} and {STAR-RIS} assisted {mMIMO} communication:
  {Analysis} and optimization,'' vol.~72, no.~9, pp. 11\,975--11\,989.

\bibitem{Huang2020}
C.~Huang \emph{et~al.}, ``Holographic {MIMO} surfaces for {6G} wireless
  networks: Opportunities, challenges, and trends,'' \emph{IEEE Wireless
  Commun.}, vol.~27, no.~5, pp. 118--125, 2020.

\bibitem{Wu2019}
Q.~Wu and R.~Zhang, ``Intelligent reflecting surface enhanced wireless network
  via joint active and passive beamforming,'' \emph{IEEE Trans. Wireless
  Commun.}, vol.~18, no.~11, pp. 5394--5409, 2019.

\bibitem{Bjoernson2019b}
E.~Bj{\"o}rnson, {\"O}.~{\"O}zdogan, and E.~G. Larsson, ``Intelligent
  reflecting surface versus decode-and-forward: {How} large surfaces are needed
  to beat relaying?'' vol.~9, no.~2, pp. 244--248.

\bibitem{Yang2020b}
Y.~Yang \emph{et~al.}, ``Intelligent reflecting surface meets {OFDM: Protocol}
  design and rate maximization,'' \emph{IEEE Trans. Commun.}, vol.~68, no.~7,
  pp. 4522--4535, 2020.

\bibitem{Zhao2020}
M.-M. Zhao \emph{et~al.}, ``Intelligent reflecting surface enhanced wireless
  networks: {Two}-timescale beamforming optimization,'' \emph{IEEE Trans.
  Wireless Commun.}, vol.~20, no.~1, pp. 2--17, 2020.

\bibitem{Mu2021}
X.~Mu \emph{et~al.}, ``Simultaneously transmitting and reflecting {(STAR) RIS}
  aided wireless communications,'' \emph{IEEE Trans. Wireless Commun.},
  vol.~21, no.~5, pp. 3083--3098, 2021.

\bibitem{Papazafeiropoulos2023}
A.~Papazafeiropoulos \emph{et~al.}, ``Achievable rate of a {STAR-RIS} assisted
  massive {MIMO} system under spatially-correlated channels,'' \emph{IEEE
  Trans. Wireless Commun.}, pp. 1--1, 2023.

\bibitem{Papazafeiropoulos2023a}
A.~Papazafeiropoulos, P.~Kourtessis, and S.~Chatzinotas, ``{Max-Min SINR}
  analysis of {STAR-RIS} assisted massive {MIMO} systems with hardware
  impairments,'' \emph{IEEE Trans. Wireless Commun.}, pp. 1--1, 2023.

\bibitem{Papazafeiropoulos2023b}
A.~Papazafeiropoulos \emph{et~al.}, ``{STAR-RIS} assisted cell-free massive
  {MIMO} system under spatially-correlated channels,'' pp. 1--16.

\bibitem{Pan2020}
C.~{Pan} \emph{et~al.}, ``Multicell {MIMO} communications relying on
  intelligent reflecting surfaces,'' \emph{IEEE Trans. Wireless Commun.},
  vol.~19, no.~8, pp. 5218--5233, 2020.

\bibitem{Ye2020}
J.~Ye, S.~Guo, and M.-S. Alouini, ``Joint reflecting and precoding designs for
  {SER} minimization in reconfigurable intelligent surfaces assisted {MIMO}
  systems,'' \emph{IEEE Trans. Wireless Commun.}, vol.~19, no.~8, pp.
  5561--5574, 2020.

\bibitem{Zhang2020a}
S.~Zhang and R.~Zhang, ``Capacity characterization for intelligent reflecting
  surface aided {MIMO} communication,'' \emph{IEEE J. Sel. Areas Commun.},
  vol.~38, no.~8, pp. 1823--1838.

\bibitem{Perovic2021}
N.~S. Perovi{\'c} \emph{et~al.}, ``Achievable rate optimization for {MIMO}
  systems with reconfigurable intelligent surfaces,'' \emph{IEEE Trans.
  Wireless Commun.}, vol.~20, no.~6, pp. 3865--3882, 2021.

\bibitem{Marzetta2016}
T.~L. Marzetta \emph{et~al.}, \emph{Fundamentals of Massive MIMO}.\hskip 1em
  plus 0.5em minus 0.4em\relax Cambridge University Press, 2016.

\bibitem{Wan2021}
Z.~Wan \emph{et~al.}, ``Terahertz massive {MIMO} with holographic
  reconfigurable intelligent surfaces,'' \emph{IEEE Trans. Commun.}, vol.~69,
  no.~7, pp. 4732--4750, 2021.

\bibitem{Hu2018}
S.~Hu, F.~Rusek, and O.~Edfors, ``Beyond massive {MIMO}: The potential of data
  transmission with large intelligent surfaces,'' \emph{IEEE Trans. Signal
  Process.}, vol.~66, no.~10, pp. 2746--2758, 2018.

\bibitem{Pizzo2020}
A.~Pizzo, T.~L. Marzetta, and L.~Sanguinetti, ``Spatially-stationary model for
  holographic {MIMO} small-scale fading,'' \emph{IEEE J. Sel. Areas Commun.},
  vol.~38, no.~9, pp. 1964--1979, 2020.

\bibitem{Pizzo2022}
A.~Pizzo, L.~Sanguinetti, and T.~L. Marzetta, ``Fourier plane-wave series
  expansion for holographic {MIMO} communications,'' \emph{IEEE Trans. Wireless
  Commun.}, vol.~21, no.~9, pp. 6890--6905, 2022.

\bibitem{Demir2022}
{\"O}.~T. Demir, E.~Bj{\"o}rnson, and L.~Sanguinetti, ``Channel modeling and
  channel estimation for holographic massive {MIMO} with planar arrays,''
  \emph{IEEE Wireless Commun. Let.}, vol.~11, no.~5, pp. 997--1001, 2022.

\bibitem{Liu2022}
C.~Liu \emph{et~al.}, ``A programmable diffractive deep neural network based on
  a digital-coding metasurface array,'' \emph{Nature Electronics}, vol.~5,
  no.~2, pp. 113--122, 2022.

\bibitem{An2023a}
J.~An \emph{et~al.}, ``Stacked intelligent metasurfaces for multiuser
  beamforming in the wave domain,'' \emph{networks}, vol.~7, p.~13, 2023.

\bibitem{An2023}
------, ``Stacked intelligent metasurfaces for efficient holographic {MIMO}
  communications in {6G},'' \emph{IEEE J. Sel. Areas Commun.}, 2023.

\bibitem{Oezdogan2020}
{\"O}.~{\"O}zdogan, E.~Bj{\"o}rnson, and E.~G. Larsson, ``Using intelligent
  reflecting surfaces for rank improvement in {MIMO} communications,'' in
  \emph{ICASSP 2020-2020 IEEE International Conference on Acoustics, Speech and
  Signal Processing (ICASSP)}.\hskip 1em plus 0.5em minus 0.4em\relax IEEE,
  2020, pp. 9160--9164.

\bibitem{Perovic2020}
N.~S. Perovi{\'c}, M.~Di~Renzo, and M.~F. Flanagan, ``Channel capacity
  optimization using reconfigurable intelligent surfaces in indoor {mmWave}
  environments,'' in \emph{ICC 2020-2020 IEEE International Conference on
  Communications (ICC)}.\hskip 1em plus 0.5em minus 0.4em\relax IEEE, 2020, pp.
  1--7.

\bibitem{Abeywickrama2020}
S.~Abeywickrama \emph{et~al.}, ``Intelligent reflecting surface: {Practical}
  phase shift model and beamforming optimization,'' \emph{IEEE Trans. Commun.},
  vol.~68, no.~9, pp. 5849--5863, 2020.

\bibitem{Lin2018}
X.~Lin \emph{et~al.}, ``All-optical machine learning using diffractive deep
  neural networks,'' \emph{Science}, vol. 361, no. 6406, pp. 1004--1008, 2018.

\bibitem{LeCun2015}
Y.~LeCun, Y.~Bengio, and G.~Hinton, ``Deep learning,'' \emph{Nature}, vol. 521,
  no. 7553, pp. 436--444, 2015.

\bibitem{Hu2022}
X.~Hu \emph{et~al.}, ``Holographic beamforming for ultra massive {MIMO} with
  limited radiation amplitudes: How many quantized bits do we need?''
  \emph{IEEE Commun. Let.}, vol.~26, no.~6, pp. 1403--1407, 2022.

\bibitem{Dai2020}
L.~Dai \emph{et~al.}, ``Reconfigurable intelligent surface-based wireless
  communications: Antenna design, prototyping, and experimental results,''
  \emph{IEEE Access}, vol.~8, pp. 45\,913--45\,923, 2020.

\bibitem{Rappaport2015}
T.~S. Rappaport \emph{et~al.}, ``Wideband millimeter-wave propagation
  measurements and channel models for future wireless communication system
  design,'' \emph{IEEE Trans. Commun.}, vol.~63, no.~9, pp. 3029--3056, 2015.

\bibitem{Nadeem2023}
Q.-U.-A. Nadeem, J.~An, and A.~Chaaban, ``Hybrid digital-wave domain channel
  estimator for stacked intelligent metasurface enabled multi-user {MISO}
  systems,'' \emph{arXiv preprint arXiv:2309.16204}.

\bibitem{ElSawy2017}
H.~ElSawy \emph{et~al.}, ``Modeling and analysis of cellular networks using
  stochastic geometry: {A} tutorial,'' \emph{IEEE Commun. Surveys Tuts.},
  vol.~19, no.~1, pp. 167--203, 2017.

\bibitem{An2023c}
J.~An \emph{et~al.}, ``Stacked intelligent metasurfaces for multiuser downlink
  beamforming in the wave domain.''

\bibitem{An2023b}
------, ``Stacked intelligent metasurface-aided {MIMO} transceiver design.''

\bibitem{Li2015b}
H.~Li and Z.~Lin, ``Accelerated proximal gradient methods for nonconvex
  programming,'' \emph{Advances in neural information processing systems},
  vol.~28, 2015.

\bibitem{hjorungnes:2011}
A.~Hj{\o}rungnes, \emph{Complex-Valued Matrix Derivatives: With Applications in
  Signal Processing and Communications}.\hskip 1em plus 0.5em minus 0.4em\relax
  Cambridge University Press, 2011.

\bibitem{Pham2018}
T.~M. Pham, R.~Farrell, and L.-N. Tran, ``Revisiting the {MIMO} capacity with
  per-antenna power constraint: Fixed-point iteration and alternating
  optimization,'' \emph{IEEE Trans. Wireless Commun.}, vol.~18, no.~1, pp.
  388--401, 2018.

\end{thebibliography}
\end{document}